\title{Working title?}
\author{Tova Brown
        \and
        Nicholas M. Ercolani
        }
\documentclass[oneside, reqno]{amsart}
\usepackage{graphicx}
\usepackage{caption, subcaption, float}

\usepackage{color}

\allowdisplaybreaks

\newtheorem{theorem}{Theorem}
\newtheorem{lemma}{Lemma}
\newtheorem{proposition}{Proposition}
\newtheorem{corollary}{Corollary}

\newcommand{\csch}{\textup{csch}}

\newcommand{\bea}{\begin{eqnarray}}
\newcommand{\eea}{\end{eqnarray}}
\newcommand{\beann}{\begin{eqnarray*}}
\newcommand{\eeann}{\end{eqnarray*}}

\begin{document}
\newpage

\title{Integrable Mappings from a Unified Perspective}

\author{Tova Brown}
\address{}
\email{tlindberg@math.arizona.edu}

\author{Nicholas M. Ercolani}
\address{Department of Mathematics, University of Arizona, 
Tucson, AZ}
\email{ercolani@math.arizona.edu}
\thanks{Supported by NSF grant DMS-1615921}
\thanks{The authors thank Joceline Lega for orignally encouraging us to study the problems considered here from a dynamical systems perspective, both theoretically and numerically, and for many helpful conversations in this regard along the way.}

\begin{abstract}
Two discrete dynamical systems are discussed and analyzed whose trajectories encode significant explicit information about a number of problems in combinatorial probability, including graphical enumeration on Riemann surfaces and random  walks in random environments. The two models are integrable and our analysis uncovers the geometric 
sources of this integrability and uses this to conceptually explain the rigorous existence and structure of elegant closed form expressions for the associated probability distributions. Connections to asymptotic results are also described. The work here brings together ideas from a variety of fields including dynamical systems theory, probability theory, classical analogues of quantum spin systems, addtion laws on elliptic curves, and links between randomness and symmetry.
\end{abstract}



\maketitle
\tableofcontents

\section{Introduction}
The modern subject of dynamical systems can often be described in terms of a number of fundamental dichotomies such as discrete versus continuous or integrable versus chaotic. The diversity that these dichotomies represent often invite comparisons across the divides that actually provide new insights into the respective extremes that can, in turn, lead to new mathematical developments. For example in the case of the former dichotomy, methods stemming from combinatorial analysis or arithmetic can be brought to bear on discrete systems that have no immediate analogue in the continuous setting but may still suggest mathematical themes that transcend the dichotomy. Such was the case, in the middle of the last century, when dynamic mappings on the interval and associated period doubling cascades led to notions of universality and phase transitions that were reminiscent of themes from statistical mechanics. Similarly, studies of chaotic systems or systems of more general ergodic type were spurred on by studying perturbations of {\it nonlinear} integrable dynamical systems, providing many concrete examples of so-called KAM phenomena with suggestions of connections to stochastic analysis.

More recently new bridges across these divides have emerged stimulated by developments in
random matrix theory, analytical combinatorics and related areas of probability theory and mathematical physics. In broad terms the realization that underlies these developments is that random settings in the presence of symmetry leads to new universality classes in the probabilistic or statistical mechanical sense; moreover, the symmetries present here frequently lead to the asymptotic distributions of these universality classes being describable in terms of integrable systems theory. A perhaps by now classical example of this is the asymptotic analysis of the longest increasing sequence in a random permutation \cite{AD99}, a long-standing problem in the area of asymptotic combinatorial probability, first posed by Erd\"os. It was solved in terms of powerful analytical methods of Riemann-Hilbert analysis arising in the theory of integrable PDE and whose universal distribution, due to Tracy and Widom, is expressible in terms of a Painlev\'e transcendent. This celebrated result has had ramifications in the study of random walks on Lie algebras\cite{BJ02}, quantum integrable systems \cite{OConn} and non-equilibrium statistical mechanics \cite{PS00, Jo98}. 

What we will discuss in this paper is certainly less extensive in scope than the results just mentioned; however, it is illustrative of these trends and it does bring together in a novel and unified way themes from combinatorial probability theory, dynamical systems theory, symmetry groups, algebraic geometry and asymptotic analysis. 

\subsection{Outline}
We will begin in Section \ref{background} with a careful description of the two dynamical systems we want to study, together with an indication of the source of their integrabilty and a numerical study of their trajectories, particularly the ones of interest to us. Then, in Section \ref{combprob}, we will reveal the combinatorial and probabilistic significance of the special orbits we have identified, as well as the hierarchy of combinatorial generating functions whose elegant closed form expressions will be the object of our analysis in the remainder of the paper.  
That analysis, in terms of elliptic function theory, its solitonic degenerations, and associated algebro-geometric features of the dynamical phase space, will be detailed in Sections
\ref{e-param1} and \ref{e-param0}. Finally in Section \ref{Conclusions} we will amplify upon the combinatorial and probabilistic relevance of our work and indicate some directions for related future research.

\section{Background} \label{background}

The explicit focus of this paper is on  two two-dimensional autonomous (that is, whose coefficients are $n$-independent) discrete dynamical systems of the form
\begin{subequations}\label{eq: discrete dynamics general c}
\begin{align}
x_{n+1} & = \frac{x_n - 1}{gx_n} - cx_n - y_n \\
y_{n+1} & = x_n,
\end{align}
\end{subequations}
equivalently given by the second-order recurrences
\begin{equation}\label{eq: recursion general c}
x_n = 1 + gx_n\left(x_{n+1} + cx_n + x_{n-1}\right),
\end{equation}
for $c$ equal to either $0$ or $1$, and with system parameter $g$. Both of these systems are integrable in a sense that will be described shortly; but despite that and the fact that superficially they look quite similar, the structure of this integrability looks quite different. However, in the remainder of this subsection we will  focus on the similarities.

An orbit, or trajectory, in either of these systems refers to the infinite sequence of points $\left\{(x_n, y_n)\right\}$ generated from an initial condition $(x_0, y_0)$ according to equations \eqref{eq: discrete dynamics general c}. Everything is considered as a function of the system parameter $g$.

Both systems have the same fixed point structure: two fixed points, one hyperbolic and the other a center. A fixed point is of the form $(x^*, x^*)$ where $x^*$ satisfies
\begin{equation}\label{eq: fixed point equation general c}
x^* = 1 + (2+c)gx^{*2}.
\end{equation}
Therefore, the two fixed points $(x^*_+, x^*_+)$ and $(x^*_-, x^*_-)$ are located at the solutions to this quadratic equation:
\begin{equation} \label{fixedpt}
x^*_{\pm} = \frac{1 \pm \sqrt{1-4(2+c)g}}{2(2+c)g}.
\end{equation}
Classification of the fixed points is determined from the Jacobian evaluated at these points:
\[ J(x^*,x^*) = \begin{pmatrix}
\frac{1}{gx^{*2}} - c & - 1 \\ 1 & 0
\end{pmatrix} = \begin{pmatrix}
\frac{1-2(c+1)gx^*}{gx^*} & -1 \\ 1 & 0
\end{pmatrix}. \]
With $\tau = \mbox{Tr}(J)$ and $\Delta = \mbox{det}(J)$, the eigenvalues of the Jacobian satisfy $|\lambda_+||\lambda_-| = \Delta = 1$ and are given by
\[ \lambda_{\pm} = \frac{\tau \pm \sqrt{\tau^2 - 4\Delta}}{2}. \]
At $(x^*_+, x^*_+)$ the determinant $\tau^2 - 4\Delta < 0$ for all appropriate values of $g$ (that is, $0 < g < \frac{1}{4(2+c)}$, in order that the fixed points are real), so both eigenvalues are complex. Therefore they are complex conjugates of each other, forcing that $|\lambda_+| = |\lambda_-| = 1$, therefore the fixed point $(x^*_+, x^*_+)$ is a center of the system. At $(x^*_-, x^*_-)$, the determinant is positive for all $g$, so the eigenvalues are real and distinct. Their magnitudes are thus not constrained to be equal but are inversely proportional, and so there are attracting and repelling directions and the fixed point $(x^*_-, x^*_-)$ is hyperbolic, a saddle.

Besides their fixed point structure, both systems also share the key property of being integrable, in the sense that each one has an integral of motion, or invariant, $I(x,y)$ satisfying $I(x_n, y_n) = I(x_{n+1}, y_{n+1})$ for all $n$. The invariants have different degrees:
\begin{equation}\label{eq: integral of motion c=1}
I_{c=1}(x,y) = xy\left(x + y - \frac{1}{g}\right) + \frac{1}{g}(x + y) - \frac{1}{g^2}
\end{equation}
and
\begin{equation} \label{eq: integral of motion c=0}
I_{c=0}(x,y) = xy(1-gx)(1-gy) + gxy - x - y + \frac{1}{g}. 
\end{equation}
Level sets $I(x,y) = e$ for constants $e$ we call ``energies'' foliate the plane with curves; an orbit of the system is restricted to lie on one such curve. The level sets are in general the real loci of elliptic curves (degree three and four curves respectively) with no singularities in the affine (finite) plane; there are also special energy values for which the curves become singular in the affine plane. Among these latter we will be primarily interested in those special values for which  the singular degeneration of the curve corresponds to an irreducible curve containing a single separatrix of the dynamical system. In Table \ref{table:DegenerateEnergies} these correspond to the energy value $e_1$. 

One should also note that there are various compactifications of the phase plane which yield an extension of the level sets to infinity where the associated curve can become singular. If one completes to the projective plane $\mathbb{P}^2$ then, in the $c =1$ case, the curve extends to be smooth at infinity. However, in the $c = 0$ case there are two singularities at infinity for all values of the energy. However, if one extends the affine plane to $\mathbb{P}^1 \times \mathbb{P}^1$ then, for $c = 0$, there are no singularities on this completion of the level curve to infinity. This will play a key role in our analysis of the $c = 0$ case with more details provided in Section \ref{e-param0}.

\begin{table}[H] 
\begin{center}
$
\begin{array}{ l | c | c } 
 & c=1 \text{ system} & c=0 \text{ system} \\ [.15em]
\hline
 & & \\
e_0 & 0  & 0  \\ [.75em]
e_1 & \dfrac{1 + 36g + (12g-1)\sqrt{1-12g}}{-54g^3} & \dfrac{1 + 20g - 8g^2 + (8g-1)\sqrt{1-8g}}{32g^2} \\ [1em]
e_2 & \dfrac{1 + 36g - (12g-1)\sqrt{1-12g}}{-54g^3} & \dfrac{1 + 20g - 8g^2 - (8g-1)\sqrt{1-8g}}{32g^2} \\ [1em]
\end{array}
$
\end{center}      
\caption{Energy values giving singular level sets of the invariant. }  \label{table:DegenerateEnergies}                    
\end{table}

\subsection{The QRT Mapping and the Sakai Classification}
There is a systematic construction that one may look to in trying to understand the integrable structure which underlies the types of dynamical systems we consider here. This construction grew out of earlier studies of classical analogues of quantum spin systems of Heisenberg type which, in its current form, is usually attributed to Quispel, Roberts and Thompson from which it derives its name, {\it the QRT mapping}. The mapping provides a construction of integral invariants of the type (\ref{eq: integral of motion c=1}) for a large class of systems of discrete Painlev\'e type. Our case of $c = 1$ is in fact the discrete Painlev\'e I equation (dPI). For full details about the QRT mapping we refer the reader to \cite{QRT89}, but for our purposes here 
we just state the general QRT form for systems of dPI type:
\bea \label{QRT}
x_{n+1} + x_{n-1} &=& - \frac{\beta x_n^2 + \epsilon x_n + \zeta}{\alpha x_n^2 + \beta x_n + \gamma}.
\eea 
The case of $c = 0$ does not fit directly within this scheme, but it does correspond to a degeneration in which the numerator and denominator on the RHS of (\ref{QRT}) have a common factor \cite{RG96}. It is not immediately clear how to find an invariant when the QRT mapping degenerates, but in our case it was possible to do so. 

QRT turns out to be the first step to an even deeper insight into the integrability of the discrete Painlev\'e systems due to Sakai \cite{Dui10} who established a correspondence between the classification of rational and elliptically fibered algebro-geometric surfaces on the one hand and Painlev\'e equations, both discrete and continuous, on the other. This correspondence encodes relations between the dynamics, B\"acklund transformations, and generalized Lie algebras. The $c = 0$ case of our study does not fit immediately into this corresponence either, but just recently a more elaborate treatement of Sakai's point configuration method was put forward \cite{KNY17}  which can be applied to this case. It will be of future interest to relate this extension to the results we present in this paper.

\subsection{Numerical study}
Numerical investigations of the systems (\ref{eq: discrete dynamics general c}) were done in Python, with simulations taking as input a value for the system parameter $g$ lying in the region $0 < g < \frac{1}{4(2+c)}$, an initial condition $(x_0, y_0)$ lying in the finite plane, and the maximum value of $n$ to which the calculation should be run. A sequence of points $\left\{(x_n, y_n)\right\}_{n=0}^{n_{max}}$ is calculated using \eqref{eq: discrete dynamics general c}, and the numbers are output as scatter plots of the orbit. Invariant level sets are viewed via contour plots.

We note that orbits near the hyperbolic fixed point experience a strong pull in the unstable directions, rapidly accumulating numerical error and producing artifacts in the image. In order to combat this source of error, a multiprecision library in Python \cite{mpmath} was used so that computations could all be done to any level of precision desired (100 decimal places of accuracy was sufficient to balance survival of an accurate simulation with computation time, in contrast to the standard approximately 13 digits of accuracy).

In all simulation images, the fixed points are plotted as black dots. The initial condition for every orbit is plotted in yellow, and all other points are plotted as blue dots. Only one orbit is shown per image. Level sets of the integral of motion are plotted as red curves. 

Due to the integrable nature of the systems, the shape of each orbit is constrained to a curve given by the level set $I(x,y) = I(x_0, y_0)$. Initial conditions which are close to one another yield orbit shapes which are close to one another, except where a separatrix lies between them due to a singularity occurring on an intermediate curve.

\subsection{Features of the $c=1$ system, dPI}
Four generic orbits are displayed in Figure \ref{fig: four orbits Rn}. We observe the presence of closed orbits encircling the center, and of three-branched (cubic elliptic curve) orbits in this system. 

\begin{figure}[H]
	\centering
	\includegraphics[width=.49\linewidth]{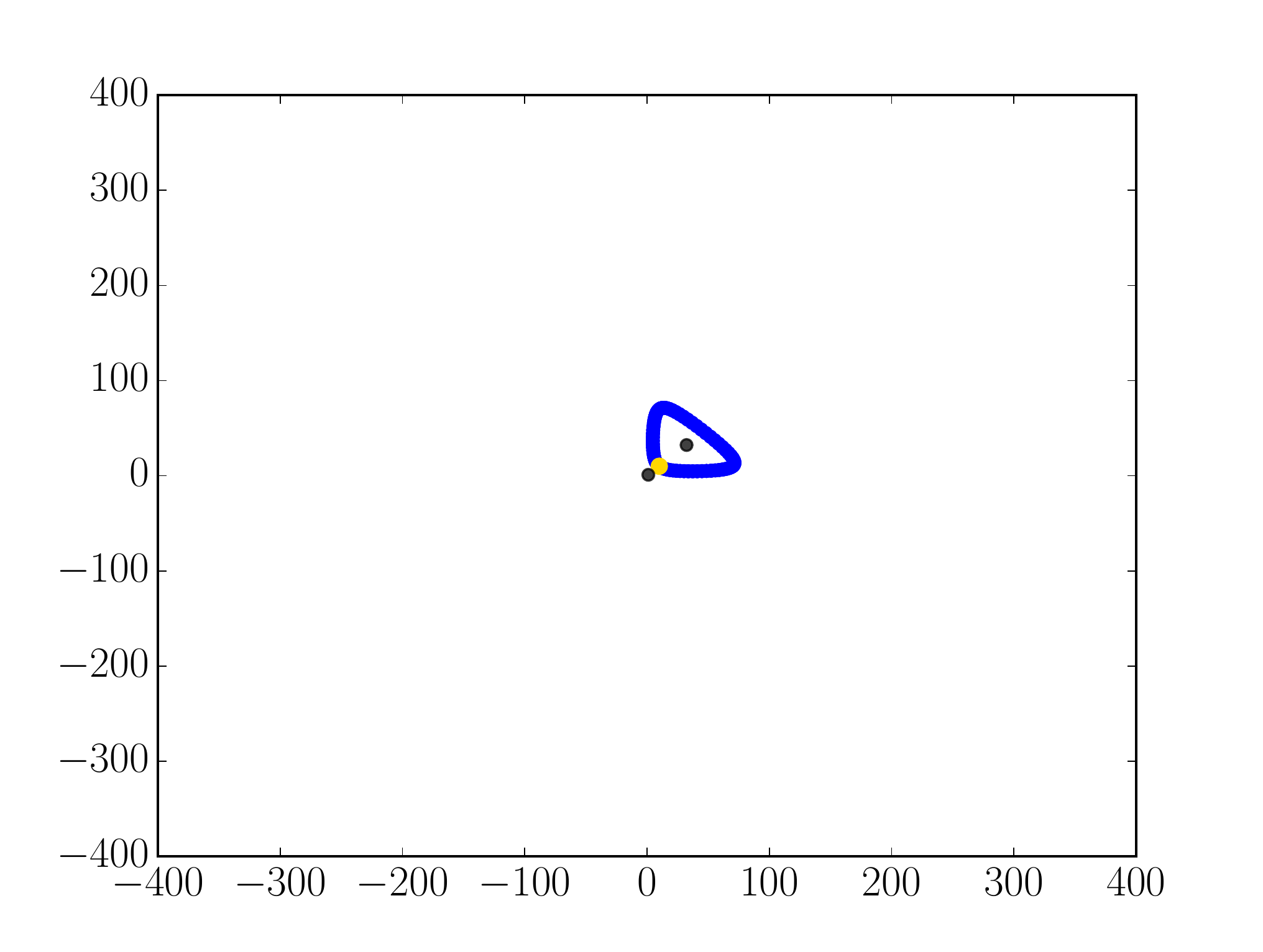}
	\includegraphics[width=.49\linewidth]{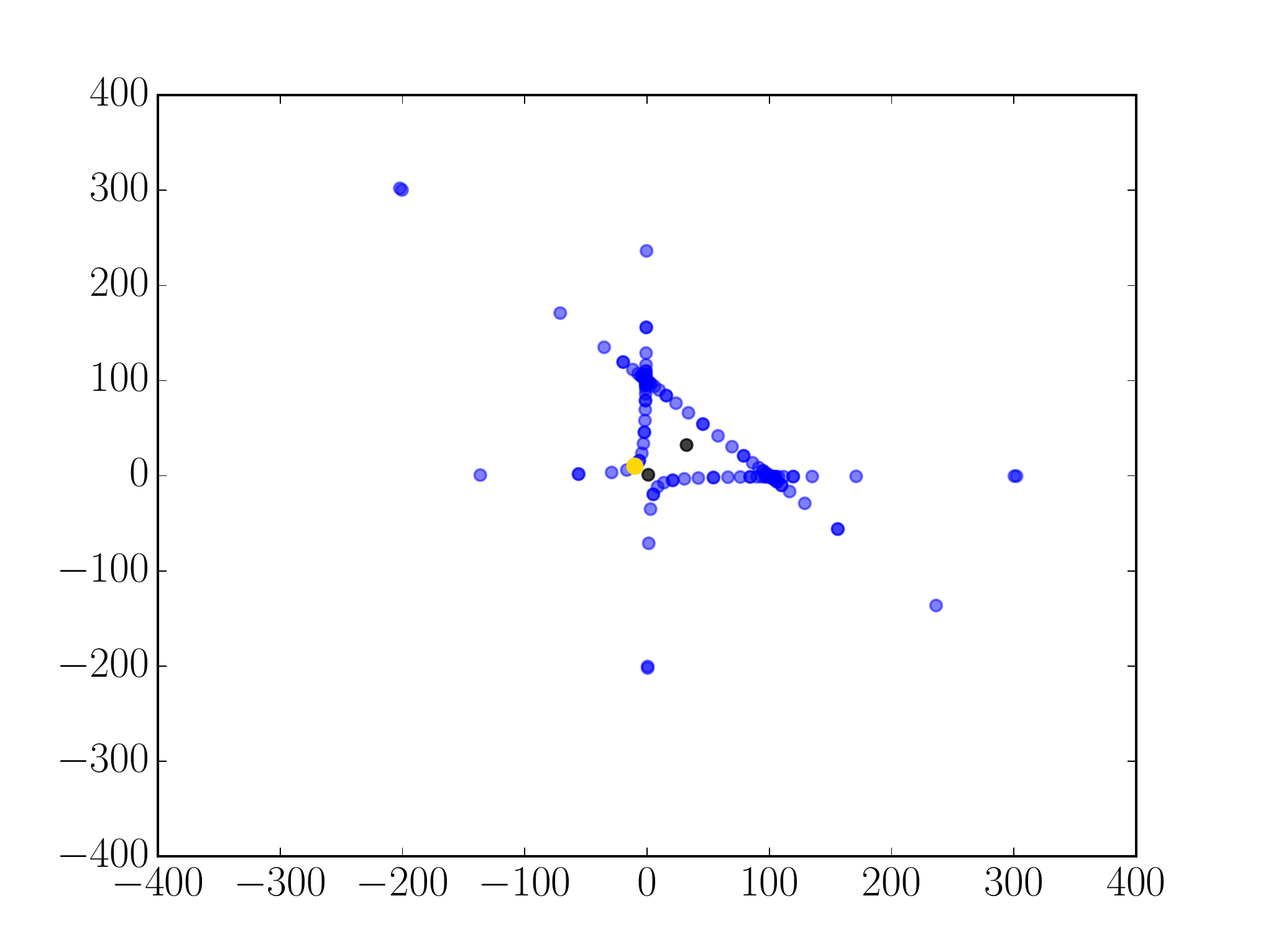} \\
	\includegraphics[width=.49\linewidth]{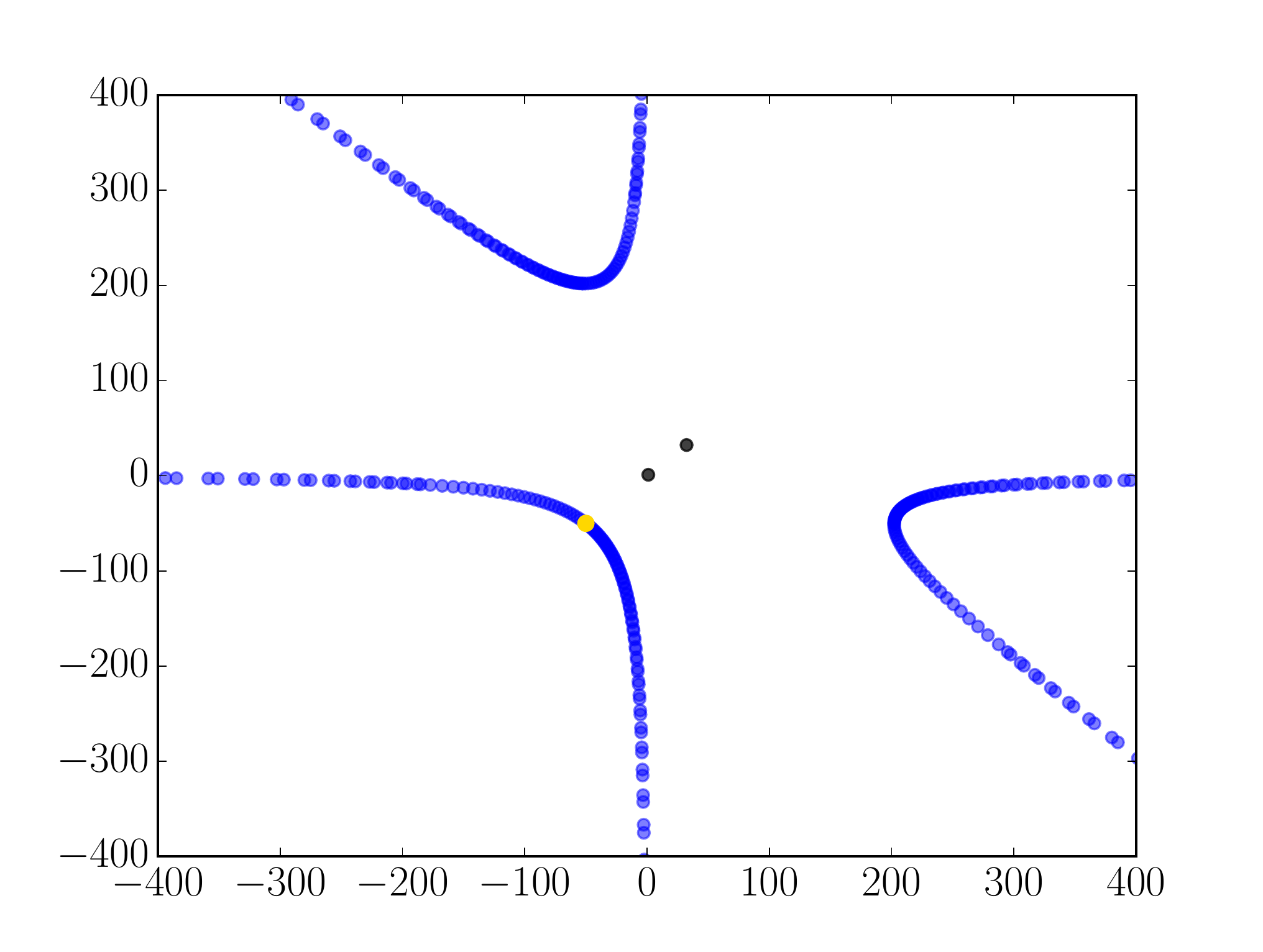}
	\includegraphics[width=.49\linewidth]{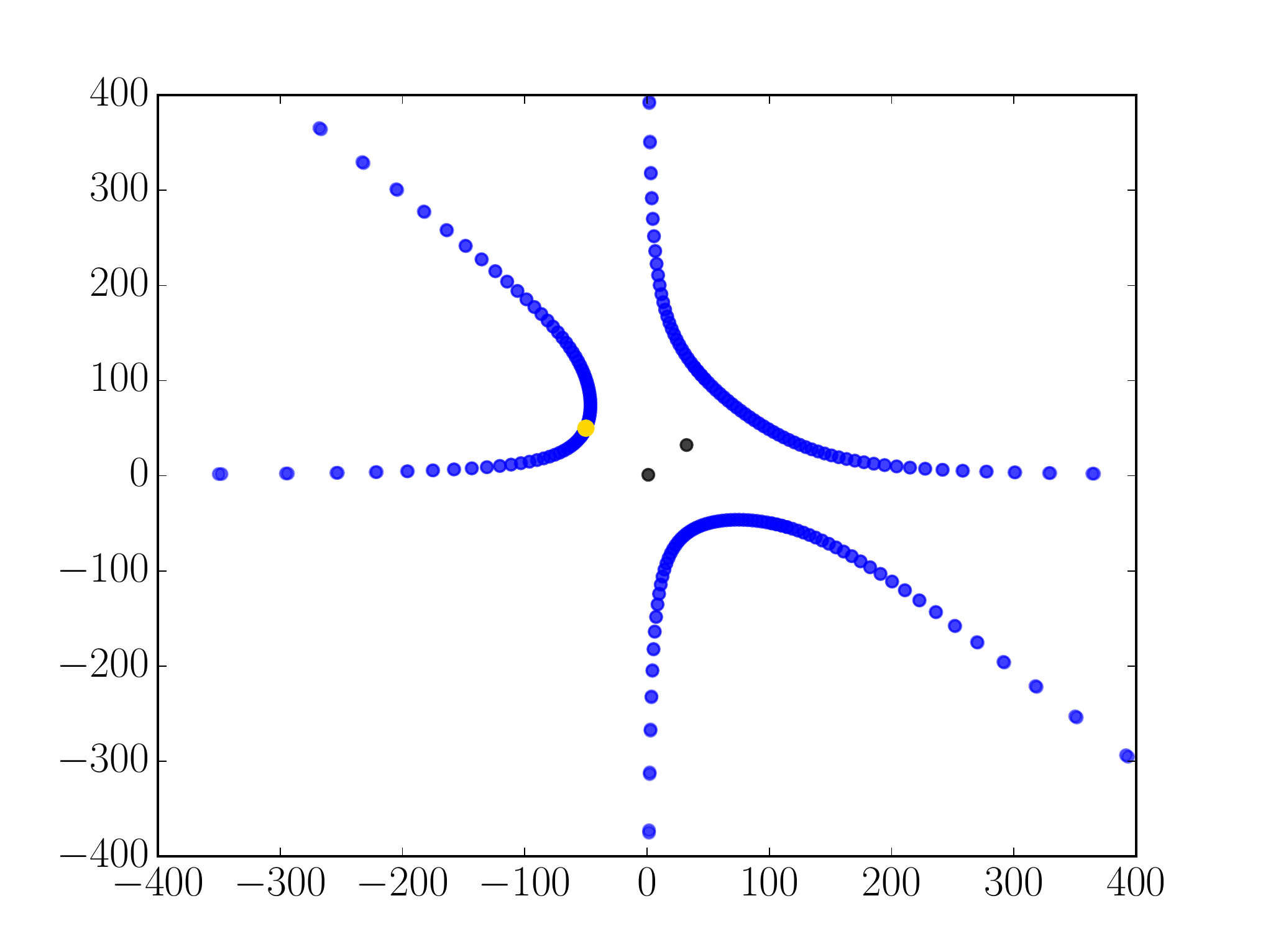}
	\caption{Four orbits of the system (blue), with their fixed points (black), and initial conditions (yellow) at $(10, 10)$, $(-10, 10)$, $(-50, -50)$, and $(-50, 50)$, listed left to right, top to bottom.}
	\label{fig: four orbits Rn}
\end{figure}
The three singular level sets of this system are shown in Figure \ref{fig: separatrices Rn}, separately and superimposed on one set of axes.

\begin{figure}[H]
	\centering
	\includegraphics[width=.49\linewidth]{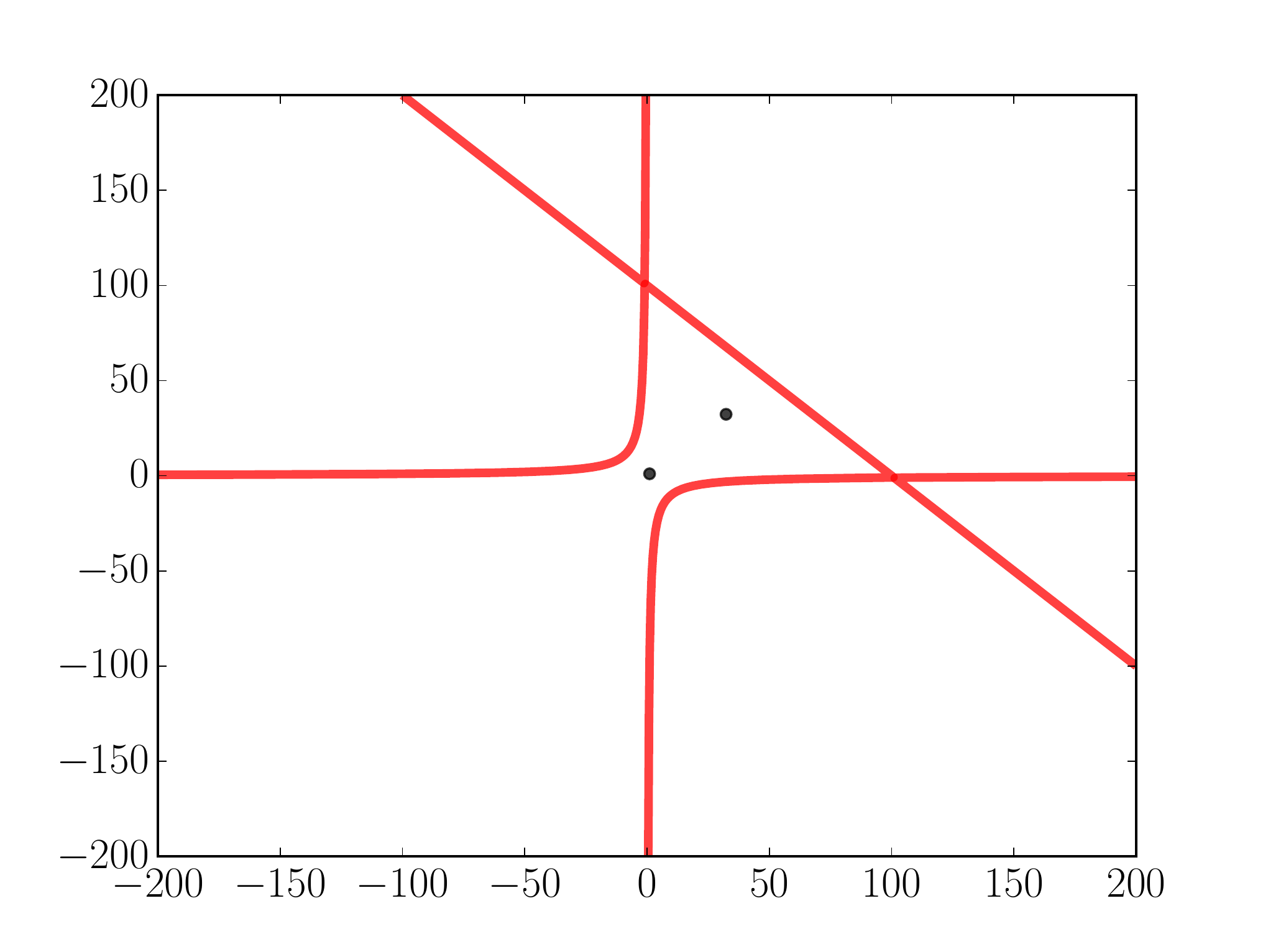}
	\includegraphics[width=.49\linewidth]{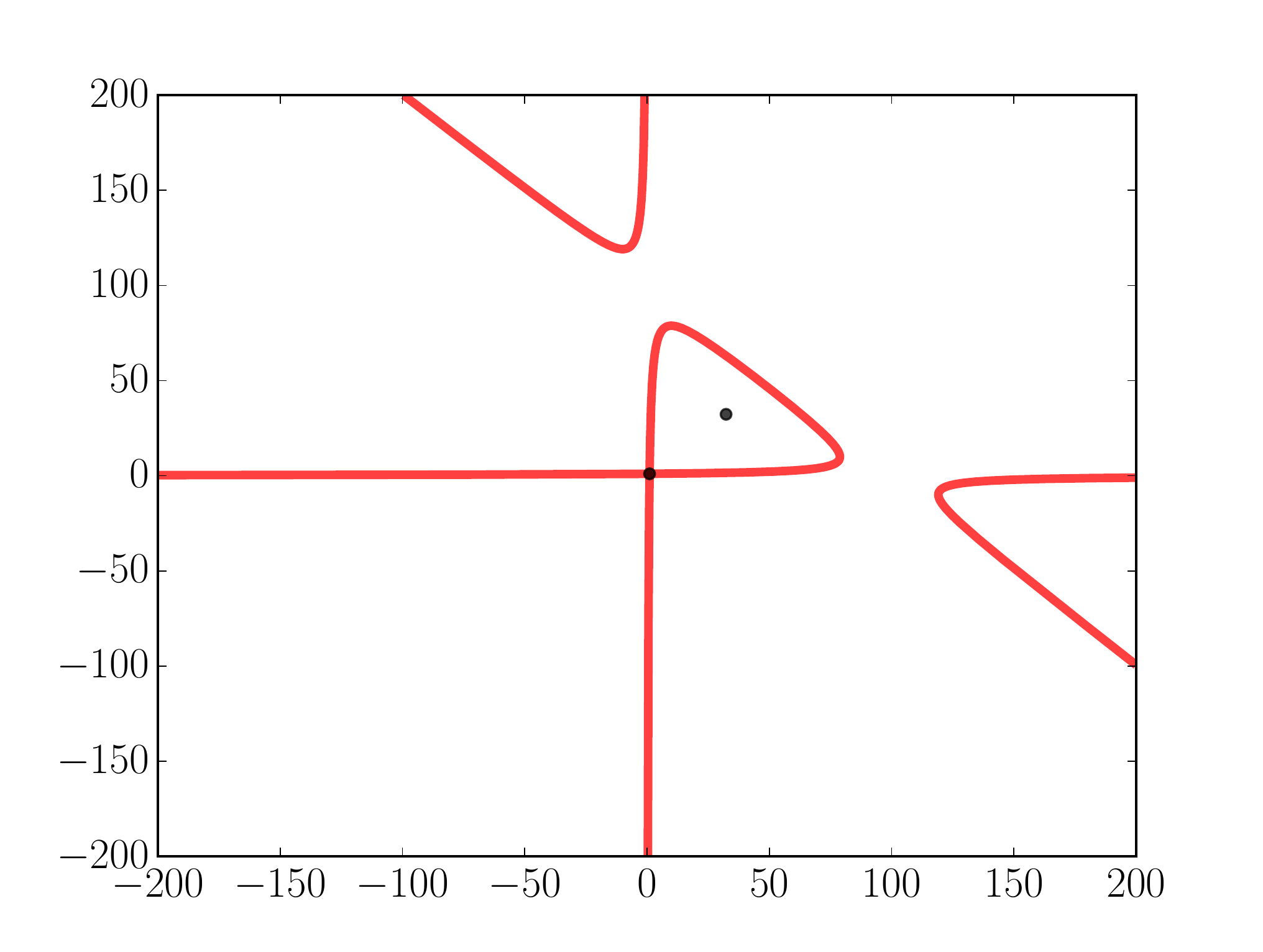} \\
	\includegraphics[width=.49\linewidth]{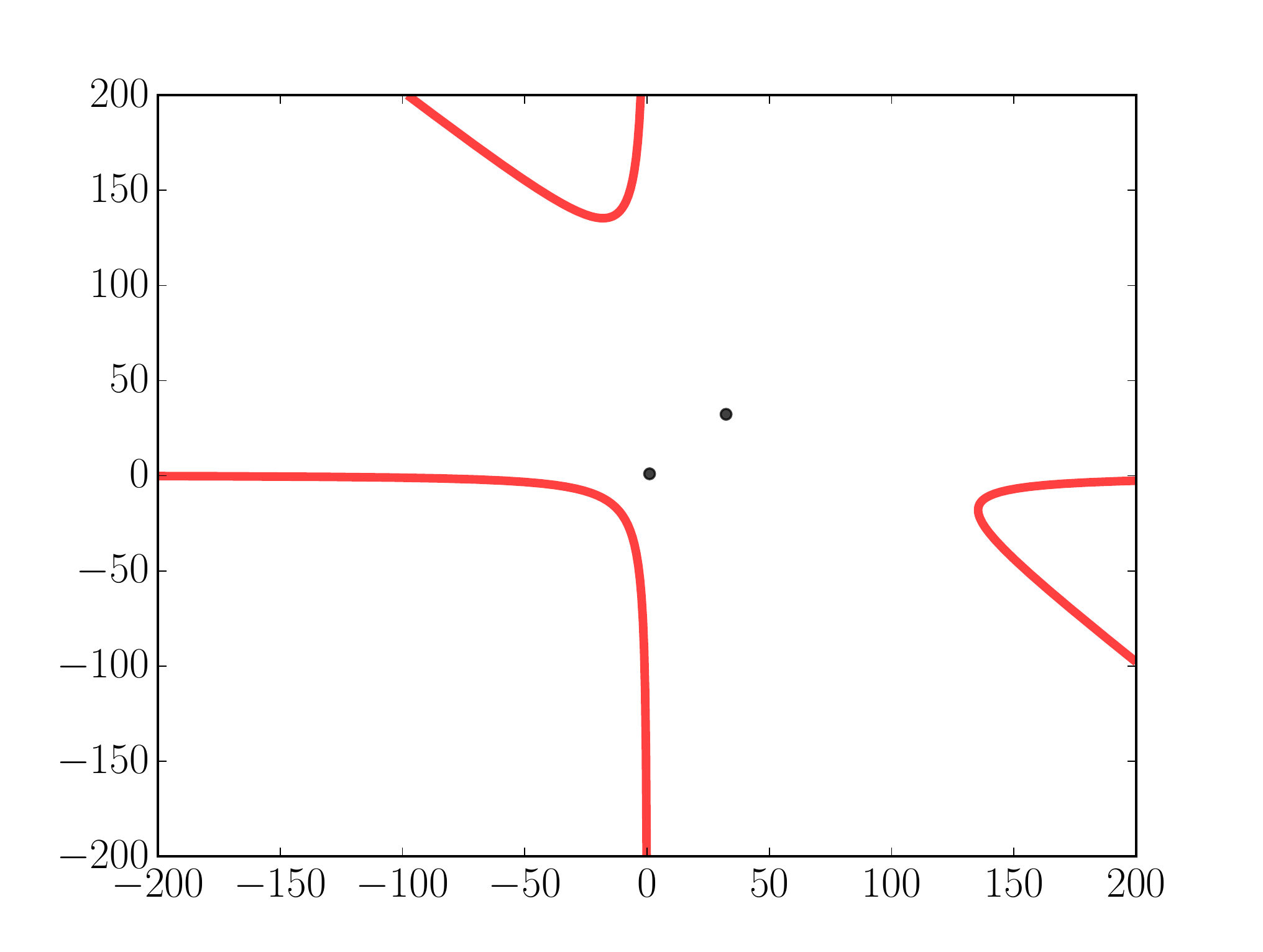}
	\includegraphics[width=.49\linewidth]{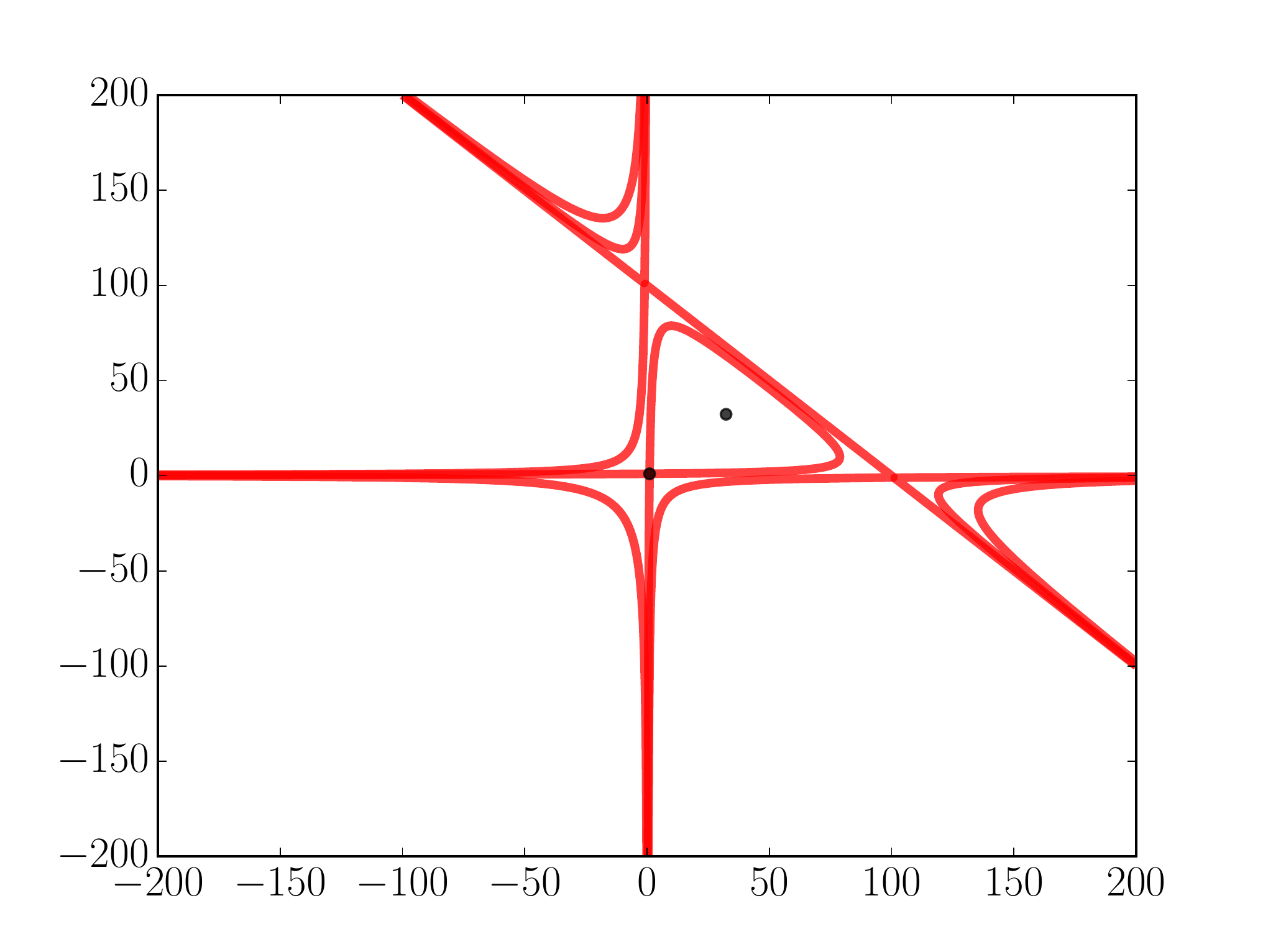}
	\caption{The three level sets $I(x,y) = 0, \approx -9897.9, \mbox{ and } \approx -40472.5$ at the degenerate energies when $g = 0.01$. The fourth image is the superposition of the others.}
	\label{fig: separatrices Rn}
\end{figure}

For sufficiently small positive values of the system parameter $g$, the qualitative shape of the entire system is preserved by continuous deformation. As $g$ approaches a critical value $g_c = \frac{1}{12} = 0.08\bar{3}$, the nearby elliptic and hyperbolic fixed points approach each other. Figure \ref{fig: change with g} illustrates this movement and the changing shape of orbits via plotting of the $e_1$ separatrix.

\begin{figure}[H]
	\centering
	\includegraphics[width=.3\linewidth]{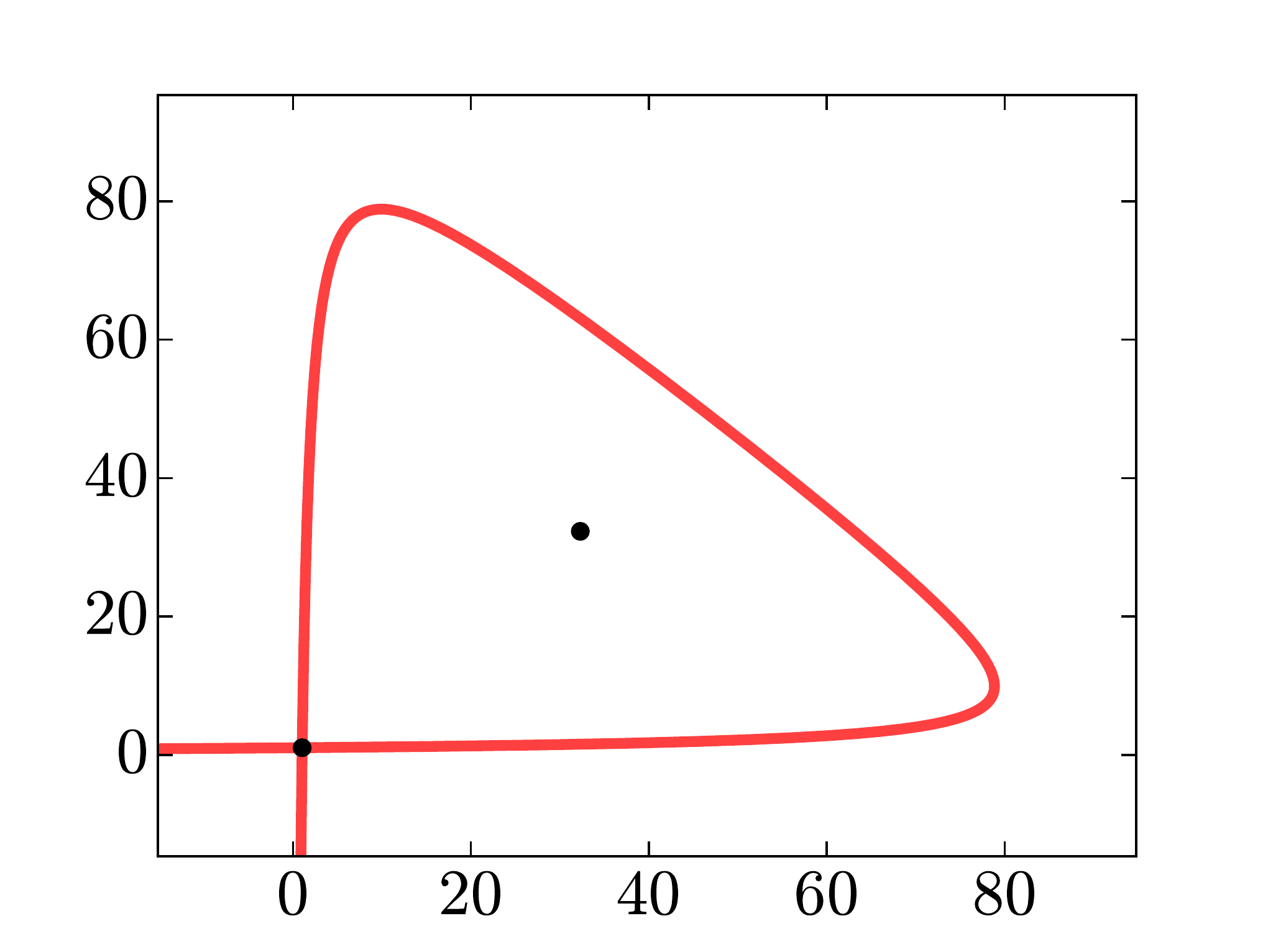}
	\includegraphics[width=.3\linewidth]{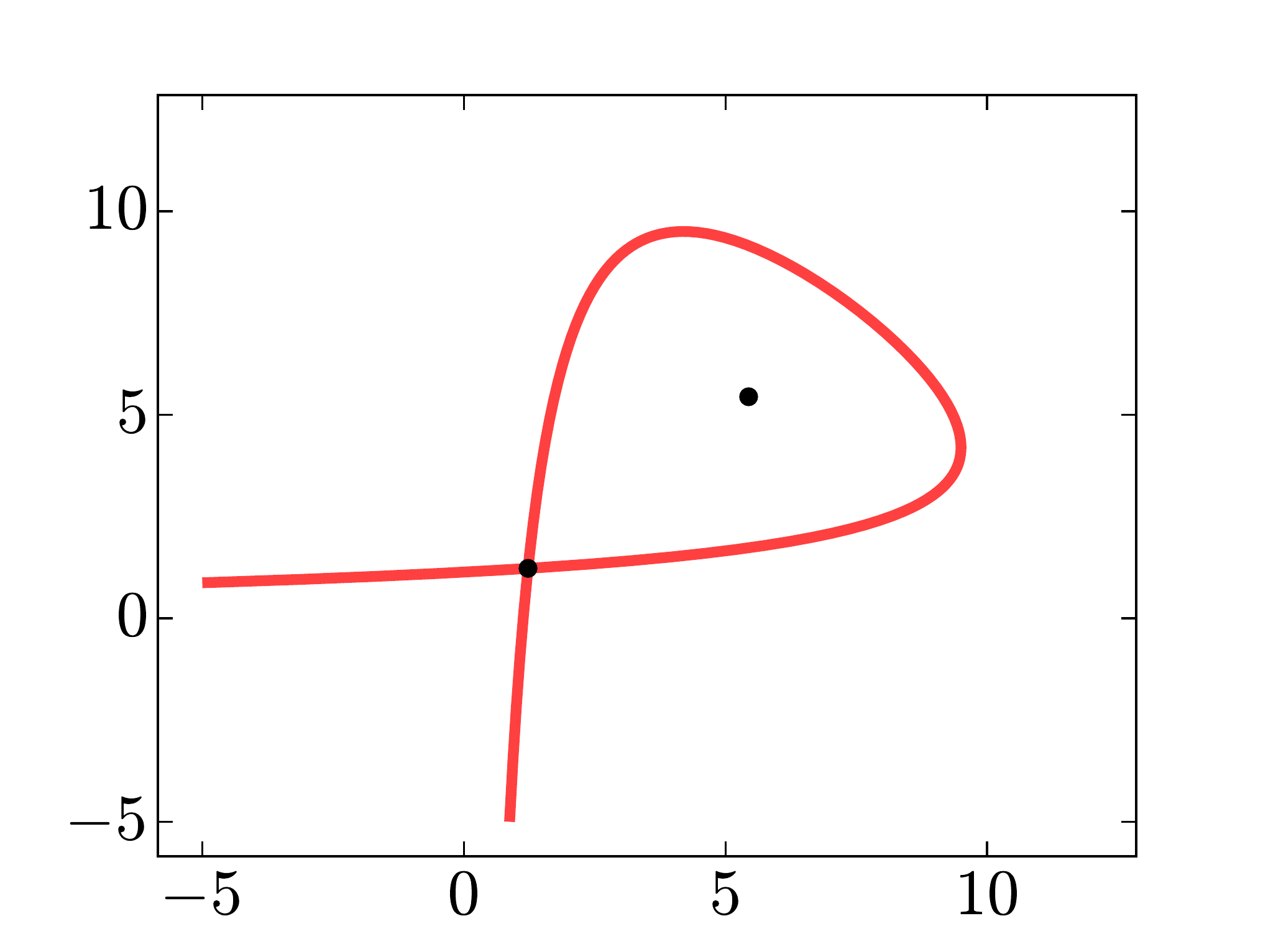}
	\includegraphics[width=.3\linewidth]{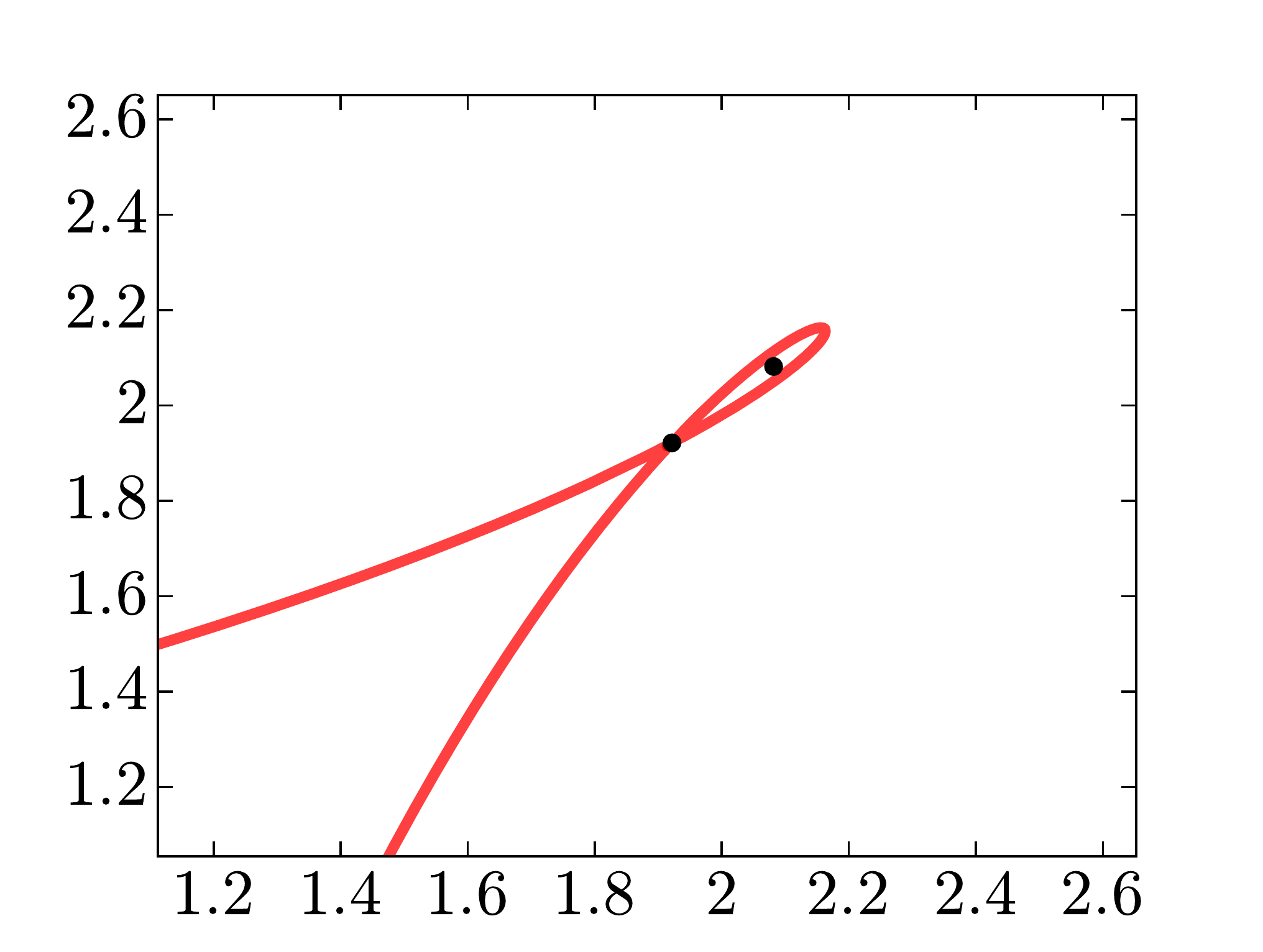}
	\caption{Three images of the $I(x,y) = e_1$ separatrix showing the change in the shape of the system as $g$ approaches $\textstyle g_c = 0.08\bar{3}$. The fixed points move closer to one another as $g \rightarrow g_c$ and coalesce when $g = g_c$. The separatrix is plotted for values $g = 0.01$, $g = 0.05$, and $g = 0.0833$ from left to right.}
	\label{fig: change with g}
\end{figure}

\subsection{Features of the $c=0$ system} In comparing the following figures with those of the previous section, we point out that, in terms of what immediately meets the eye, the respective trajectories show many similarities with the exception of the four-branched nature of this system, in contrast to the three-branched nature of the $c=1$ system.

\begin{figure}[H]
	\centering
	\includegraphics[width=.49\linewidth]{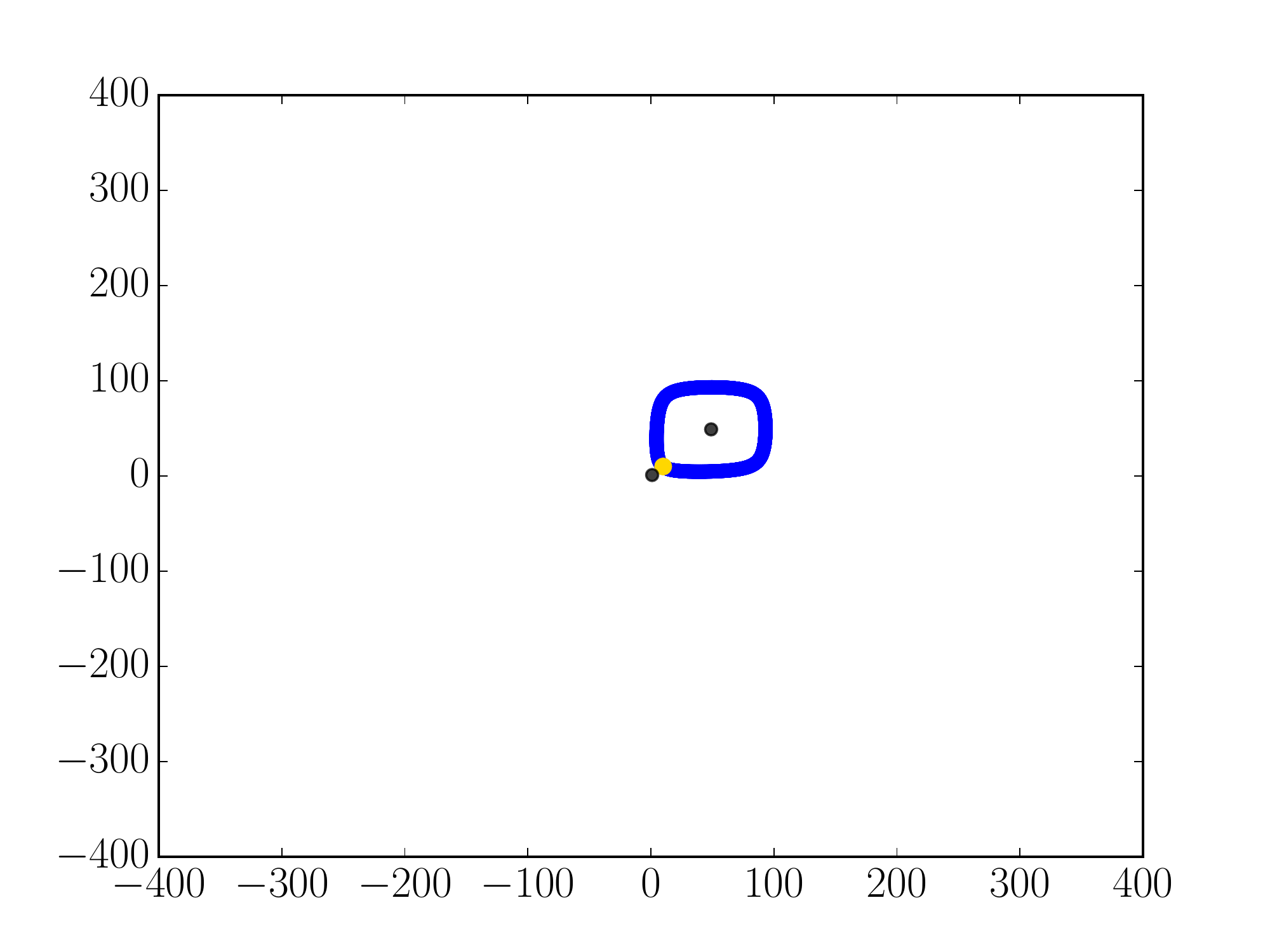}
	\includegraphics[width=.49\linewidth]{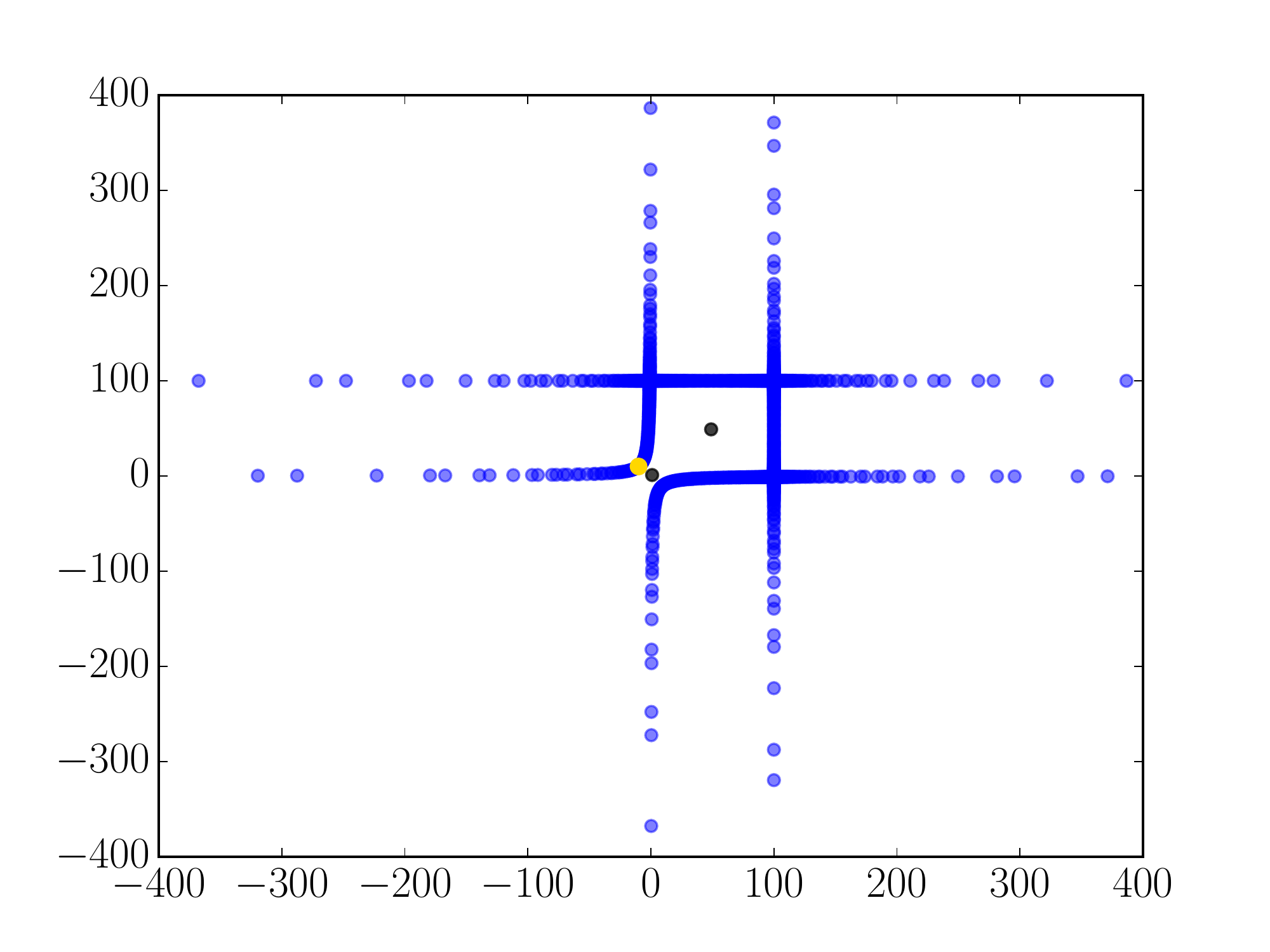} \\
	\includegraphics[width=.49\linewidth]{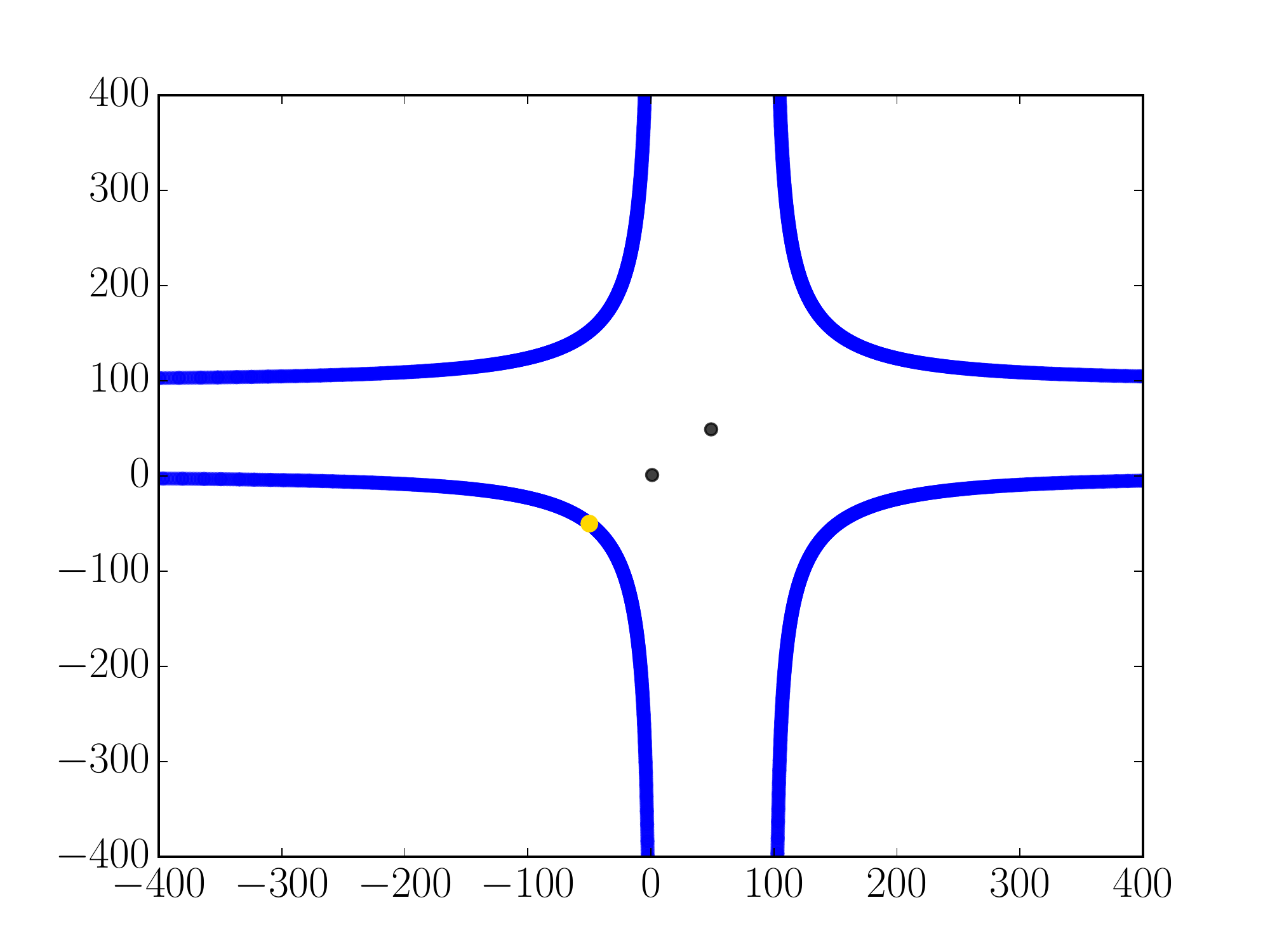}
	\includegraphics[width=.49\linewidth]{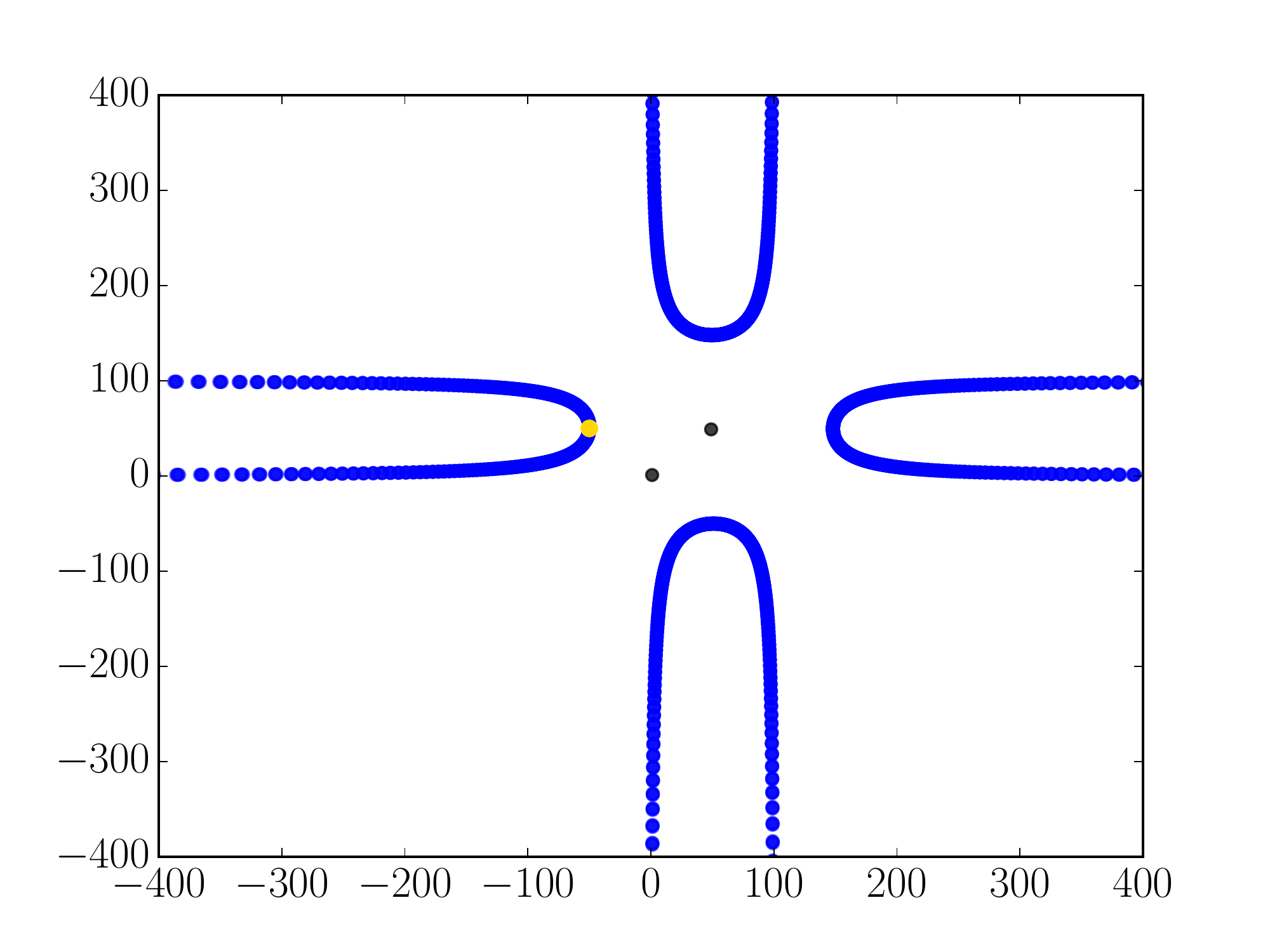}
	\caption{Four orbits of the system (blue), with their fixed points (black), and initial conditions (yellow) at $(10, 10)$, $(-10, 10)$, $(-50, -50)$, and $(-50, 50)$, listed left to right, top to bottom.}
	\label{fig: four orbits Tj}
\end{figure}

\begin{figure}[H]
	\centering
	\includegraphics[width=.49\linewidth]{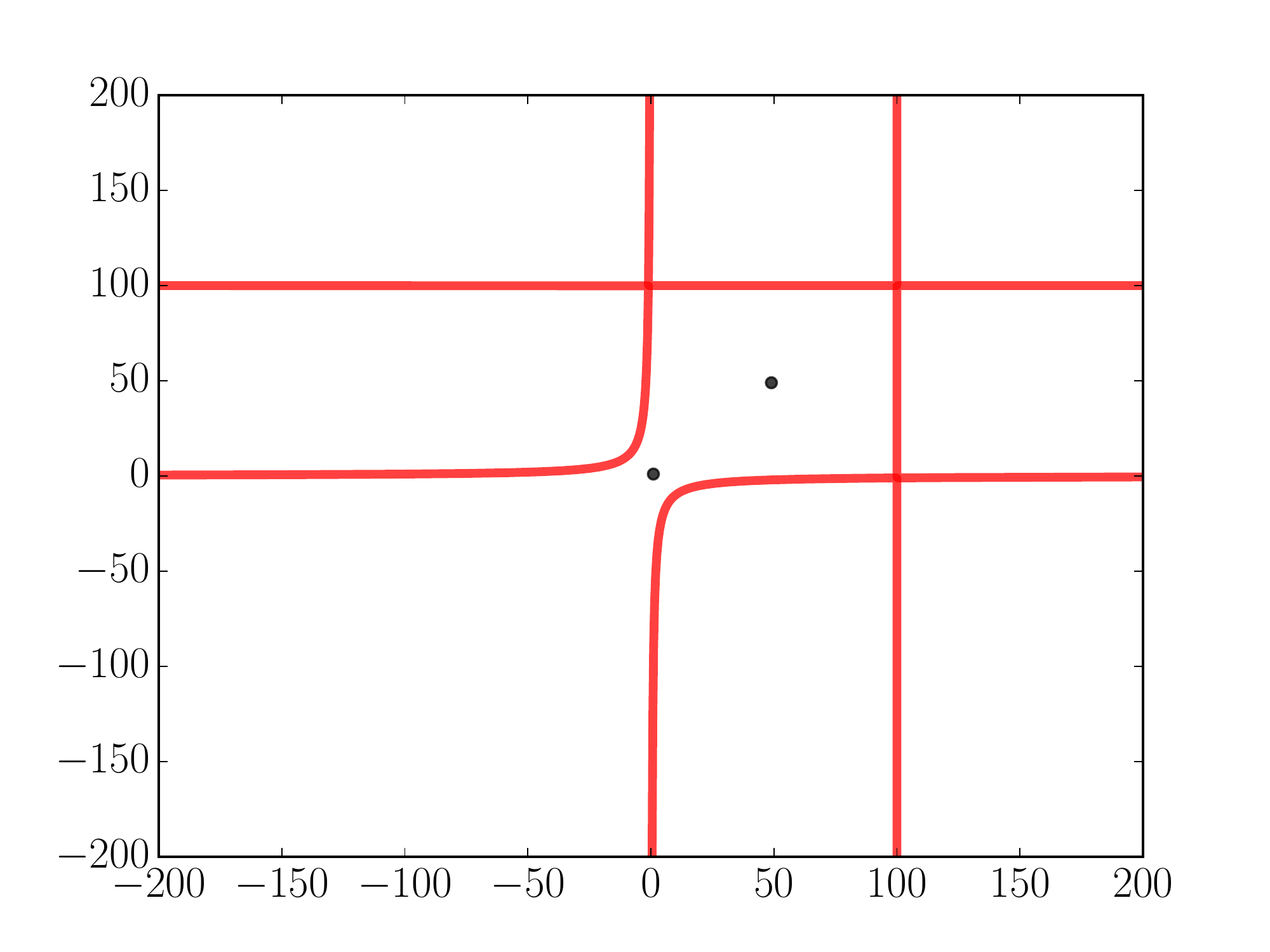}
	\includegraphics[width=.49\linewidth]{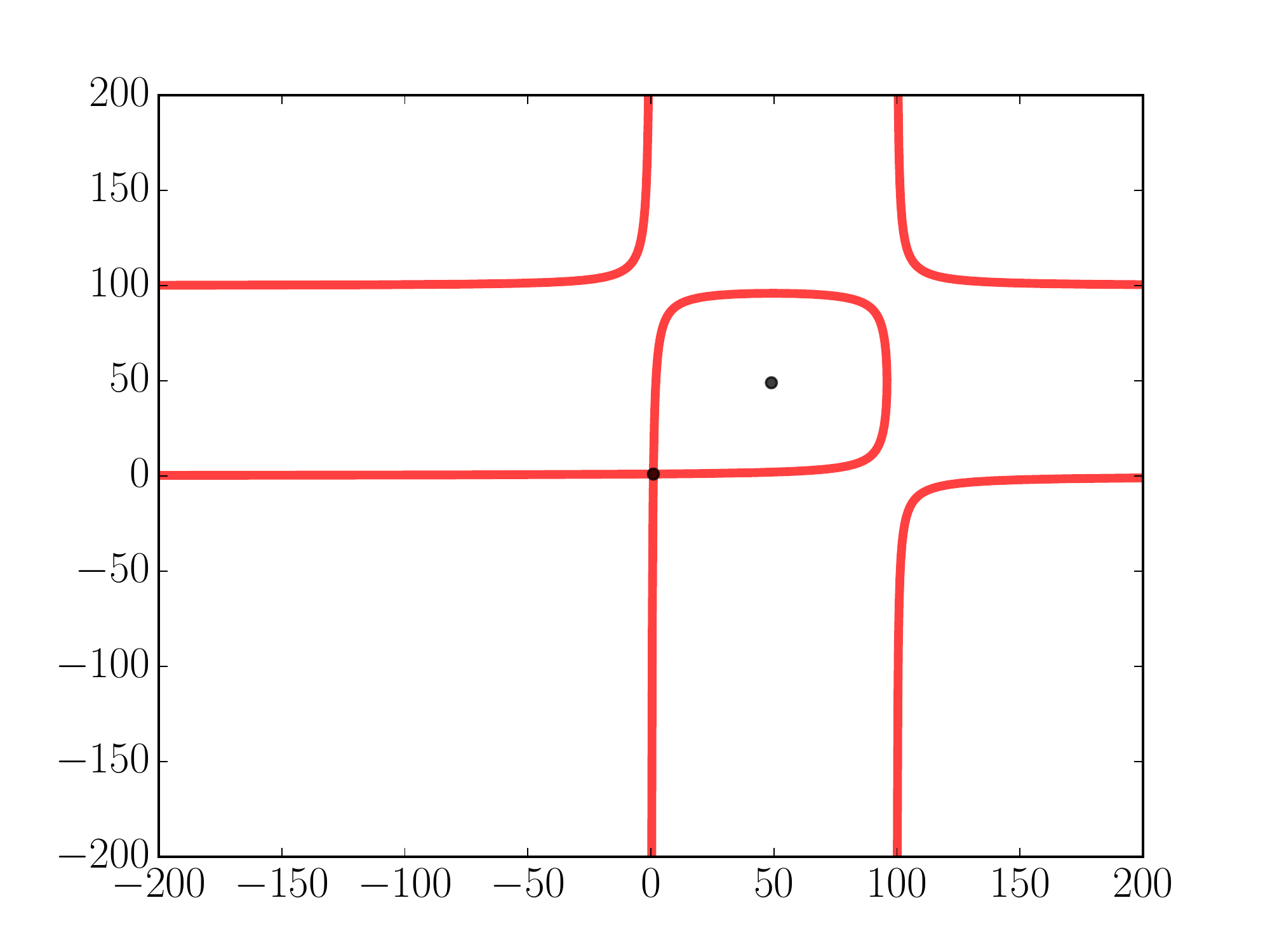} \\
	\includegraphics[width=.49\linewidth]{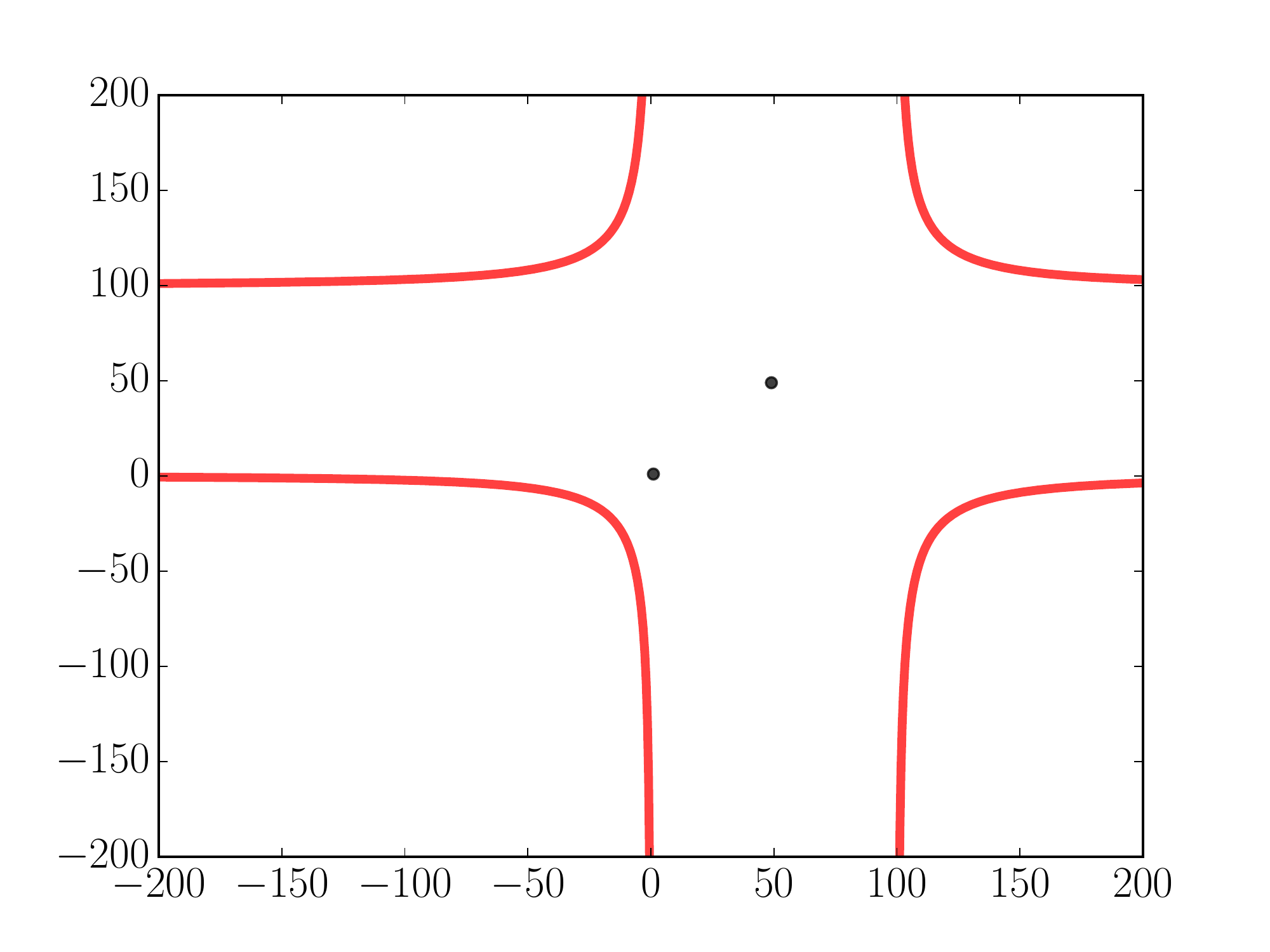}
	\includegraphics[width=.49\linewidth]{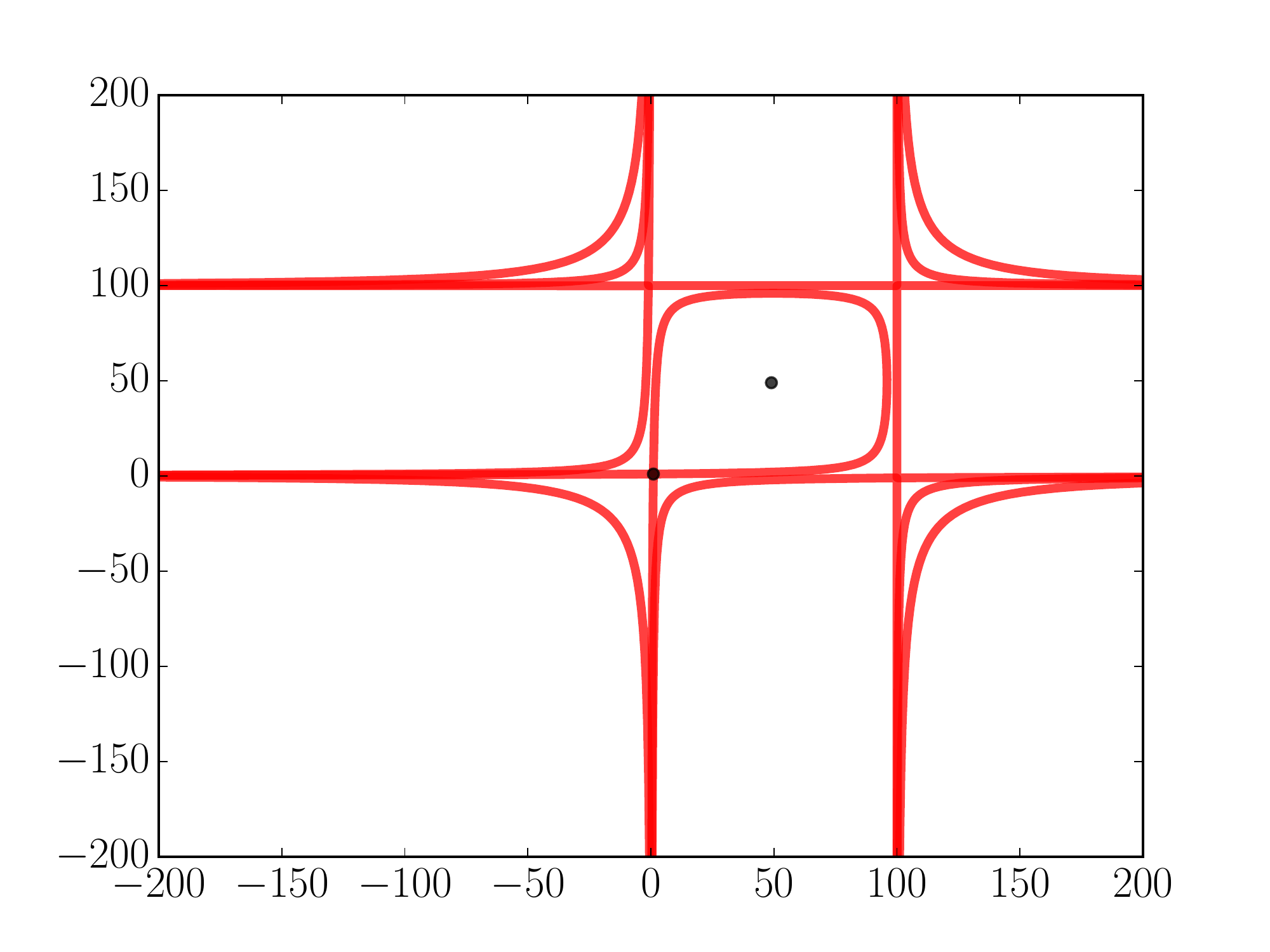}
	\caption{The three level sets $I(x,y) = 0, \approx 98.9897, \mbox{ and } \approx 650.5103$ at the degenerate energies when $g = 0.01$. The fourth image is the superposition of the others.}
	\label{fig: separatrices Tj}
\end{figure}

\begin{figure}[H]
	\centering
	\includegraphics[width=.3\linewidth]{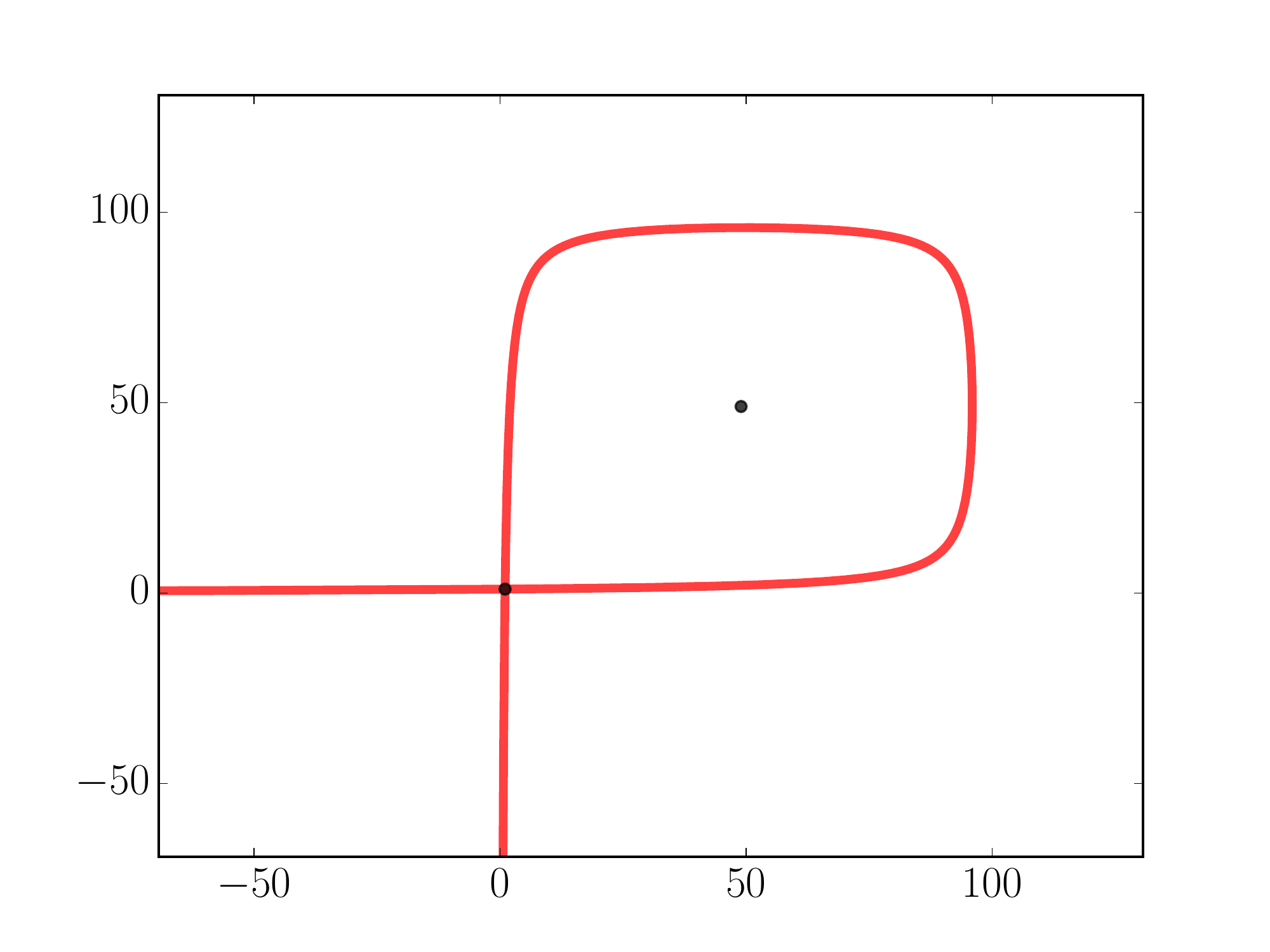}
	\includegraphics[width=.3\linewidth]{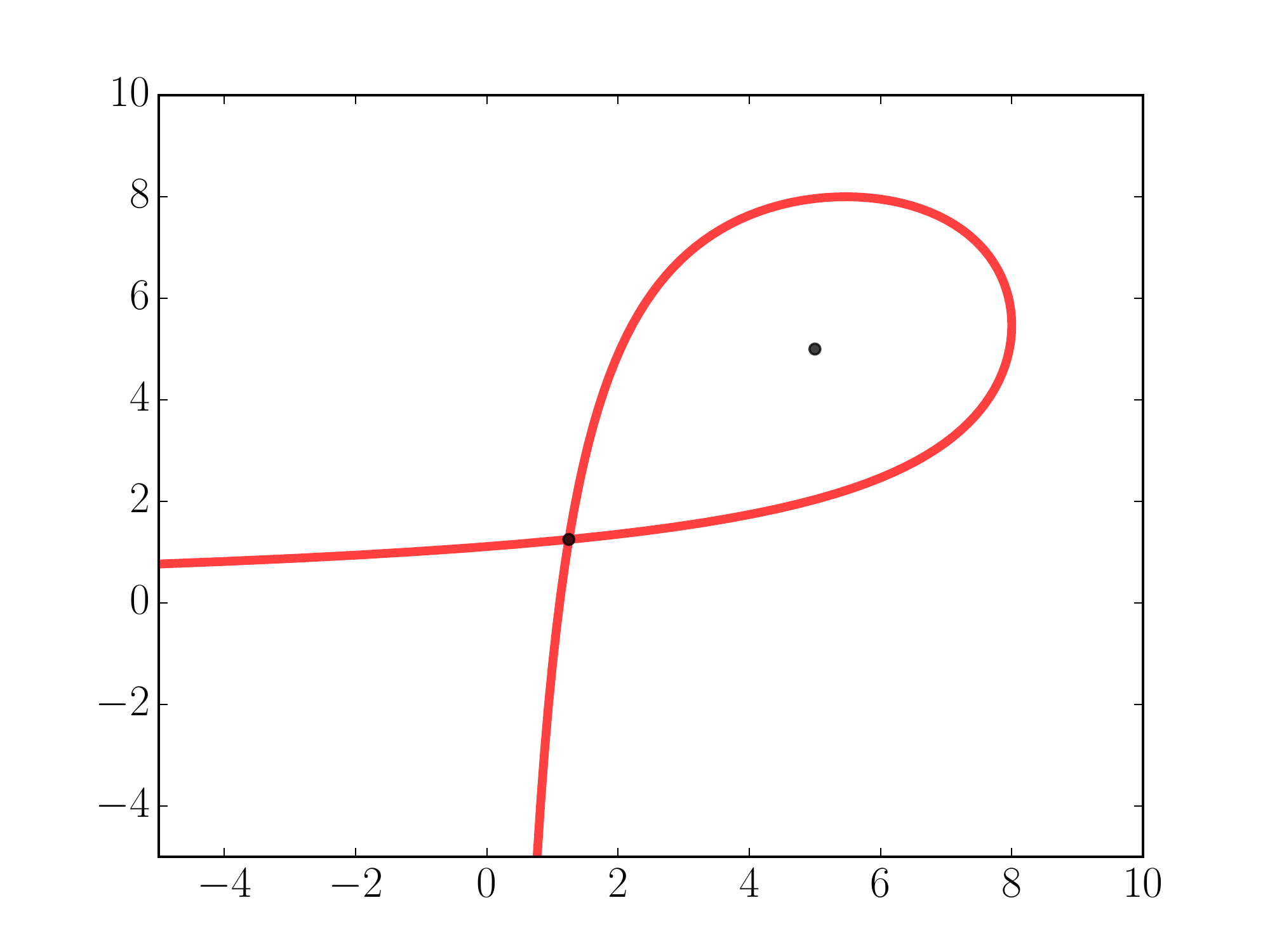}
	\includegraphics[width=.3\linewidth]{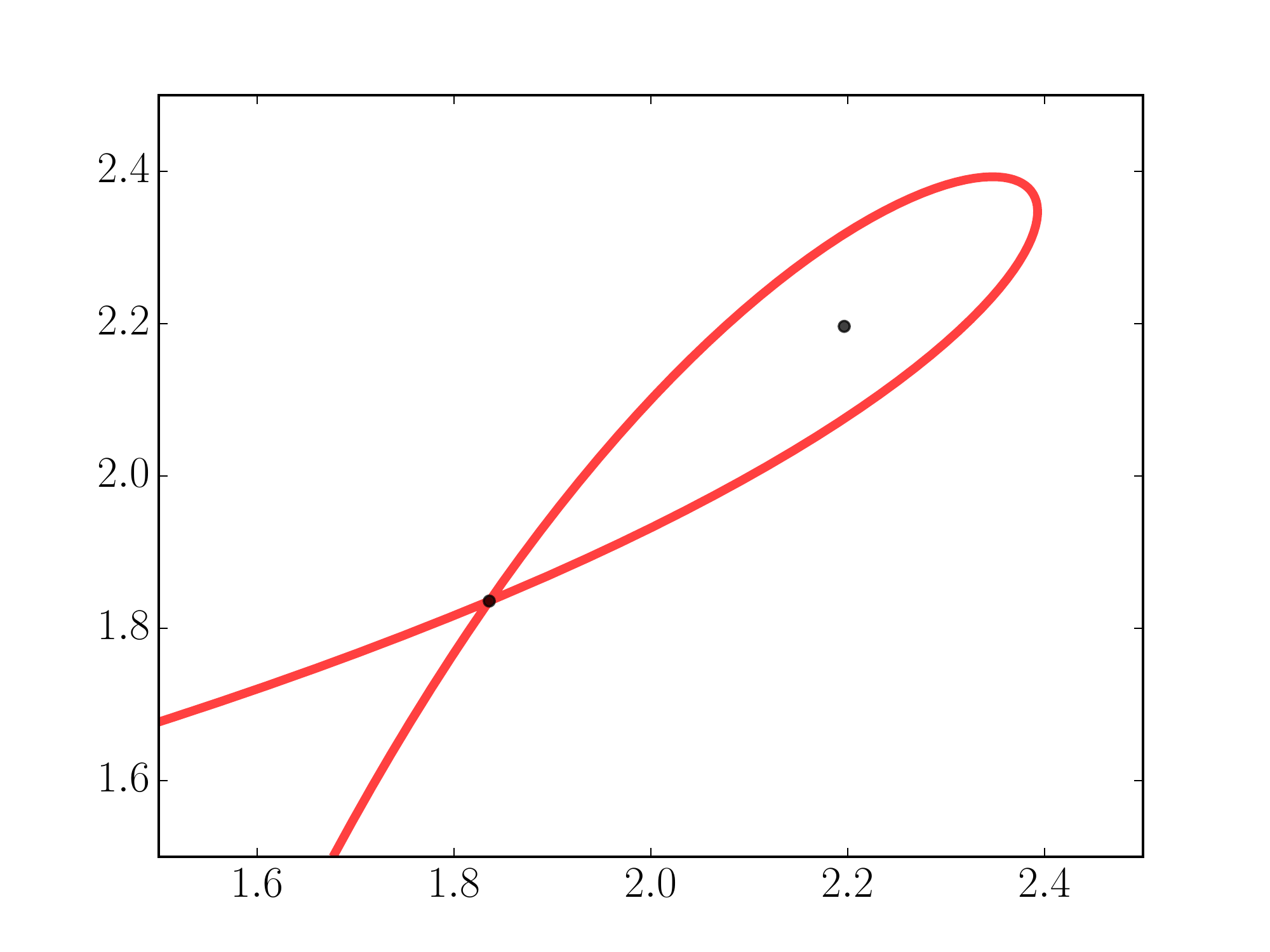}
	\caption{Three images of a separatrix showing the change in the shape of the system as $g$ approaches $\textstyle g_c = 0.125$. The fixed points move closer to one another as $g \rightarrow g_c$ and coalesce when $g = g_c$. The separatrix is plotted for values $g = 0.01$, $g = 0.08$, and $g = 0.124$ from left to right.}
	\label{fig: change with g Tj}
\end{figure}


\section{Combinatorial and probabilistic applications} \label{combprob}
In this section, we introduce a primary motivation for studying the dynamical systems described in the previous sections. The two systems are directly related to two combinatorial problems coming from the world of map enumeration. A \emph{map} is a graph which is embedded into a surface so that (i) the vertices are distinct points in the surface, (ii) the edges are curves on the surface intersecting only at vertices, and (iii) if the surface is cut along the graph, what remains is a disjoint union of connected components, called \emph{faces}, each homeomorphic to an open disk. By \emph{graph} we mean a collection of vertices, edges, and incidence relations, with loops and multi-edges allowed, taken to be connected unless stated otherwise. By \emph{surface} we mean a compact connected complex-analytic manifold of complex dimension one, up to orientation-preserving homeomorphism.

Map enumeration has a history dating back to the early 1960s, which is briefly summarized  in section \ref{sec: comb context}. We begin in section \ref{sec: comb 2 problems} by introducing two problems and solutions appearing in the literature which motivated the present study. These problems are basic in the sense that they treat {\it triangulations} and {\it quadrangulations}, that is, maps of the lowest face degrees. They are interesting because both families of maps have a notion of a distance, and the enumeration knows about this distance, opening up opportunities for questions about random metrics, statistical mechanics on random lattices, random walks in random environments and random combinatorial structures in general. These problems are special in that they have very nice closed form combinatorial solutions which exhibit hints of integrability, something that was not expected during their initial study but is turning out to characterize a whole collection of problems related to map enumeration  (see section \ref{sec: comb context}). 

\subsection{Two combinatorial problems}\label{sec: comb 2 problems}

\subsubsection{Quadrangulations}
The material in this section is from the seminal paper \cite{BDG03}, where this problem was first introduced and a solution was proposed. It concerns the enumeration of 4-valent 2-legged (that is, there are also two 1-valent vertices) planar maps, sorted according to a notion of geodesic distance that will be defined presently. It is equivalent to counting rooted planar quadrangulations (that is, the face degree, rather than the vertex degree, is 4) with an origin vertex, again counting according to an appropriate distance \cite{BDG03a}.

Let $R = R(g)$ denote the generating function for the family of 2-legged 4-valent planar maps:
\[ R(g) = \sum_{k\ge 0} r_k g^k \]
where $r_k$ denotes the number of 2-legged 4-valent planar maps on $k$ vertices. By Schaeffer's bijection \cite{Sch97}, this is also the generating function for 4-valent blossom trees. Blossom trees are rooted trees whose external vertices (called leaves) are either black or white, and in which a certain number of leaves at every internal vertex are black, depending on the valence of the trees (1 in a 4-valent blossom tree). The subtree structure, shown in Figure \ref{fig: environment of the root}, is that each internal vertex of a 4-valent $k$-vertex blossom tree has exactly one descendent black leaf and two descendent subtrees on $k-1$ vertices, where a single white leaf is a subtree with zero vertices. The number of ways to choose two such subtrees and arrange them in the ordered tree is
\[ r_k = 3 \sum_{i=0}^{k-1} r_i r_{k-1-i}. \]
Therefore, the generating function $R(g)$ satisfies
\begin{equation}\label{eq: R functional equation ch2}
R = 1 + 3gR^2
\end{equation}
with analytic solution
\begin{equation}
R(g) = \frac{1 - \sqrt{1-12g}}{6g}.
\end{equation}

\begin{figure}[h]
\centering
\begin{subfigure}{.33\textwidth}
  \centering
  \includegraphics[width=.5\linewidth]{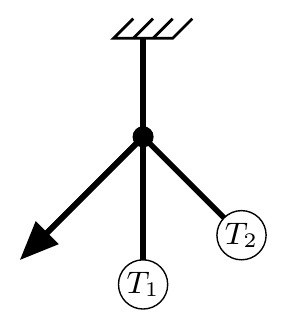}
\end{subfigure}%
\begin{subfigure}{.33\textwidth}
  \centering
  \includegraphics[width=.5\linewidth]{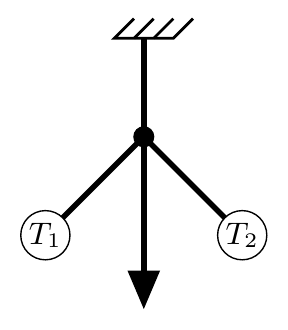}
\end{subfigure}
\begin{subfigure}{.33\textwidth}
  \centering
  \includegraphics[width=.5\linewidth]{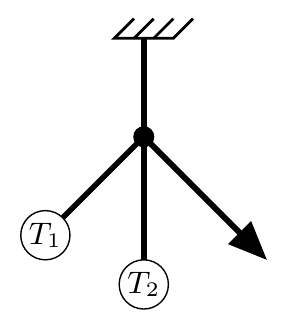}
\end{subfigure}
\caption{Decomposition of a 4-valent blossom tree at its root: the single black leaf has three possible positions, and the other two descendants are blossom sub-trees.}
\label{fig: environment of the root}
\end{figure}

The \emph{geodesic distance} of a 2-legged map is the minimal number of edges crossed by a path on the surface connecting the two legs. Schaeffer's bijection was extended in \cite{BDG03} to keep track of the geodesic distance, preserved through the bijection and called \emph{root depth} on the tree side, of even-valent blossom trees. The subtree structure decomposition in Figure \ref{fig: environment of the root} was again applied to recursively build a blossom tree from its subtrees, and the root depth of the pieces and the whole kept track of using contour walks; the resulting generating function is for maps with an upper bound on the geodesic distance. Furthermore, defining
\[ R_n(g) := \sum_{k=0}^{\infty} r_{n,k} g^k \]
for $r_{n,k} = \#\{\mbox{2-leg 4-valent $k$-vertex planar maps with goedesic distance $\le n$}\}$
the derivation yields a second-order three-term recurrence relation for the generating functions instead of the quadratic functional relation that determined $R(g)$:
\begin{equation}\label{eq: R_n DiF recursion}
R_n = 1 + gR_n\left( R_{n-1} + R_n + R_{n+1} \right).
\end{equation}
The recurrence relation has two boundary conditions: 
\begin{subequations}\label{eq: Rn rec derivation ic}
\begin{align}
R_{-1} & = 0\label{eq: Rn rec derivation ic R-1} \\
\lim_{n\rightarrow\infty}\limits R_n & = R\label{eq: Rn rec derivation ic lim},
\end{align}
\end{subequations}
the initial condition stating that there are no maps with negative geodesic distance, and the limit condition stating that the large $n$ limit amounts to removing the upper bound on geodesic distance.

The recurrence relation (\ref{eq: R_n DiF recursion}) with boundary conditions (\ref{eq: Rn rec derivation ic}) was solved in \cite{BDG03}. The closed-form solution to the recursion is given by
\begin{equation}\label{eq: R_n cf DiF}
R_n = R \frac{(1-x^{n+1})(1-x^{n+4})}{(1-x^{n+2})(1-x^{n+3})}.
\end{equation}
where $x$ is the solution with modulus less than 1 to the characteristic equation
\begin{equation}\label{eq: x functional equation}
x + \frac{1}{x} = \frac{1-4gR}{gR}.
\end{equation}
The original paper did not contain a proof of this result; it can be directly checked that the closed-form solves the recursion, and the authors provided the formal means by which they discovered this amazing formula. The first proof of this closed-form appeared in \cite{BG12}, by continued fraction expansions and exploiting a remarkable connection to the problem of enumerating planar maps with a boundary. This leads to an expression for the generating functions in terms of symplectic Schur functions. Our approach provides an alternative proof that 
emphasizes the connection to integrable dynamical systems and natural extensions to Painlev\'e equations and Riemann-Hilbert analysis (see section \ref{sec: comb context}); however, it will be very interesting to study relations between these two approaches.

\subsubsection{Triangulations} \label{triangulations}
The second problem comes from the paper \cite{Bou06a} and is an enumeration problem of labeled plane trees which are bijectively equivalent to certain triangulations of the sphere (see section \ref{conclusions2}). A \emph{plane tree} is a planar map with only one face, and the enumeration is of the family of labeled plane trees, rooted with the root labeled 0 and with the labels of two adjacent vertices differing by $\pm 1$. The valence of these trees is unrestricted. Figure \ref{fig: bmtrees} shows a tree in this family, with its root labeled in boldface. 
Such trees are also called \emph{embedded}, because trees labeled in this way can be naturally embedded into a one-dimensional integer lattice, with the labeling of each vertex dictating its location in the lattice. 
The maximal label, or equivalently the largest point in the lattice having positive mass, is called the \emph{label height}. The label height is the appropriate notion of distance for this enumeration problem, corresponding to the geodesic distance of the previous problem.

\begin{figure}[h]
  \centering
  \includegraphics[width=.25\textwidth]{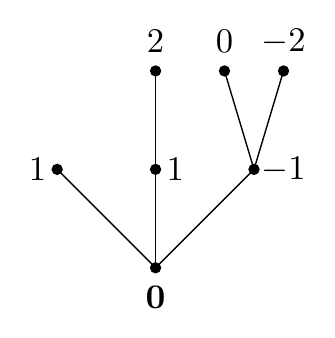}
  \caption{A labeled plane tree rooted at 0; reproduction of an example in \cite{Bou06a}.}
  \label{fig: bmtrees}
\end{figure}

Define the generating functions
\[ T(g) = \sum_{k=0}^{\infty} t_k g^k \hspace{.15in} \mbox{and} \hspace{.15in} T_j(g) = \sum_{k=0}^{\infty} t_{j,k} g^k \]
where $t_k = \#\{\mbox{trees in the family with $k$ edges}\}$ and $t_{j,k} = \# \{ \mbox{trees in the family} \linebreak \mbox{having $k$ edges and label height $\le j$} \}$. 
A rooted tree either has zero edges (and is simply the root), or it has at least one edge; in the latter case, it can be decomposed (see Figure \ref{fig: BMdecompT}) into two subtrees connected by an edge incident to the root, and the sum of the edges in the two subtrees will be one less than the number of edges in the overall tree. The subtree which doesn't contain the root is still in the family, by shifting all its labels by $+1$ or $-1$.
\begin{figure}[H]
\begin{center}
\includegraphics[width=.5\textwidth]{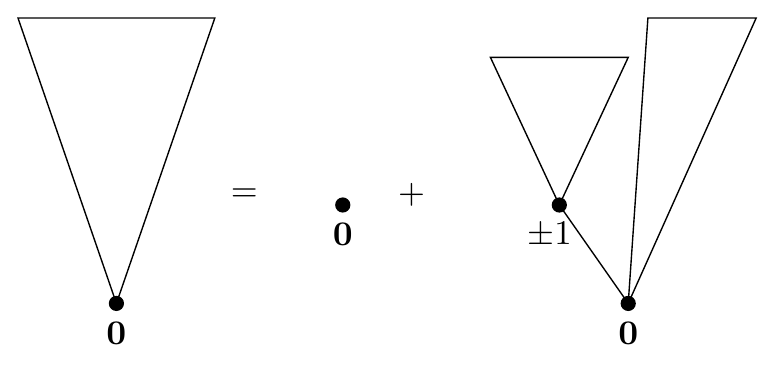}
\end{center}
\caption{Recursive breakdown at the root of a labeled plane tree: a tree is either simply a root or it can be decomposed as discussed above.}
\label{fig: BMdecompT}
\end{figure}
\noindent The recurrence relation for $t_k$, 
\[ t_k = 2\sum_{j=0}^{k-1} t_j t_{k-1-j}, \]
is summed to see that the generating function for this family of trees is also given by a quadratic equation:
\begin{equation}
T = 1 + 2gT^2.
\end{equation}

Keeping track of label height through this enumeration, the recursion relation obeyed by $T_j$ is 
\begin{equation}\label{eq: T_j recurrence}
T_j = 1 + gT_j \left( T_{j-1} + T_{j+1} \right)
\end{equation}
with boundary conditions
\begin{subequations}\label{eq: Tj rec derivation ic}
\begin{align}
& T_{-1} = 0\label{eq: Tj rec derivation ic T-1} \\
& \lim_{j\rightarrow\infty}\limits T_j = T\label{eq: Tj rec derivation ic lim}.
\end{align}
\end{subequations}

Then \cite{Bou06a} contains the following closed-form solution (formally derived, and verifiable by insertion into the recursion) for this family of generating functions: 
\begin{equation}\label{eq: T_j closed form}
T_j = T\frac{(1-Z^{j+1})(1-Z^{j+5})}{(1-Z^{j+2})(1-Z^{j+4})},
\end{equation}
where $Z \equiv Z(g)$ is implicitly given by
\bea\label{Zeqn}
x_* = \frac{(1+Z)^2}{1+Z^2},
\eea
where $x_* = x^*_-$ as defined by (\ref{fixedpt}) which also specifies its $g$-dependence.

A discrete integral of motion
\[ f(x,y) = xy(1-gx)(1-gy) + gxy - x - y \]
satisfying $f(T_{n-1}, T_n) = f(T_n, T_{n+1})$ with convergence of $T_n$ to its limit $T$ assured by $f(T_n, T_{n+1}) = f(T, T)$ was given in \cite{BDG03b}.

\subsection{Combinatorics context and applications}\label{sec: comb context}
Systematic enumeration of families of maps began with Tutte in \cite{Tut62a, Tut62b, Tut62c, Tut63}. He studied rooted maps (maps having a distinguished directed edge) and, by decomposing maps at their root, derived functional equations for generating functions which are often multi-variate due to the combinatorial complexity of the maps. Tutte's Quadratic Method along with Lagrange Inversion gives the asymptotic behavior of many families of maps (see \cite{BR86} for a summary). This foundation has made the study of map enumeration a signature component of the modern theory of random combinatorial structures. Further research has followed two methodological streams, one algebraic and the other analytic. In the later category 
one has the subject of analytic combinatorics amply illustrated in \cite{FS09} on the one hand 
and methods of Riemann-Hilbert Analysis coming from integrable systems theory that solved the Erd\"os problem on longest increasing sequences in permutations mentioned earlier. However we believe the most promising developments will come from a combination of both methods. The problems discussed in this paper are a perfect example of that. On the algebraic side was the breakthrough result of Schaeffer \cite{Sch97}, already mentioned, which built upon pioneering work by Cori and Vacquelin \cite{CV81} to establish a bijection between the families of planar quadrangulations/triangulations and appropriate families of trees. On the analytic side was the asymptotic analysis of families of orthogonal polynomials which motivated the postulation in  \cite{BDG03} of formulas like \eqref{eq: R_n cf DiF}. The rigorous asymptotic analysis of the Riemann-Hilbert problem for orthogonal polynomials with exponential weights that has been carried out in \cite{EM03, EMP08, EP12, Er11} provides a rigorous background for the analysis and extension of \eqref{eq: R_n cf DiF}. Indeed our original motivation was to understand {\it non-autonomous} extensions of  \eqref{eq: discrete dynamics general c} that arose from the analysis in the above cited papers. More will be said about that in section \ref{conclusions1}. Here we just close by noting that besides providing a natural derivation of the formulas  \eqref{eq: R_n cf DiF} and \eqref{eq: T_j closed form}, the results in this paper provide for a rigorous global analysis of these generating functions as well as the determination of their maximal domains of holomorphy. 

\section{Elliptic Parametrization, c=1} \label{e-param1}
\subsection{The Generic Orbit}
In this section we show how the combinatorial problem, corresponding to $c = 1$, embeds into a more general integrable dynamical system and provide an elliptic parametrization of the general orbits of this system. This is a discrete analogue of action-angle variables of continuous integrable systems. The approach used here is primarily that of classical function theory using canonical coordinates based on the Weierstrass $\wp$ function. In this perspective the separation of variables provided by  \eqref{HJ} may be viewed as a discrete analogue of the Hamilton-Jacobi equation. In the next theorem we will make use of such an elliptic parametrization for generic orbits of the dPI system that has already been developed by \cite{BTR01}. Our main results concern the degeneration of this parametrization and are presented in the next subsection.

\begin{theorem}\label{thm: elliptic parametrization dPI}
The discrete dynamical system
\begin{subequations}
\begin{align}
x_{n+1} & = \frac{1}{g} - \frac{1}{gx_n} - x_n - y_n\\
y_{n+1} & = x_n
\end{align}
\end{subequations}
with invariant
\[ I(x,y) = xy\left(x+y-\frac{1}{g} \right) + \frac{1}{g}(x+y) - \frac{1}{g^2} \]
satisfying $I(x_{n+1}, x_n) = I(x_n, x_{n-1})$ has an elliptic parametrization through the Weierstrass $\wp$-function as
\begin{subequations}\label{eq: x, y parametrized through wp}
\begin{align}
x(\xi) & = \frac{1}{2} \left( \frac{1}{g} + \frac{\wp'(v) - \wp'(\xi)}{\wp(v) - \wp(\xi)} \right)\label{eq: x parametrized through wp} \\
y(\xi) & = \frac{1}{2} \left( \frac{1}{g} + \frac{\wp'(v) + \wp'(\xi)}{\wp(v) - \wp(\xi)} \right)\label{eq: y parametrized through wp}
\end{align}
\end{subequations}
where $v$ is such that $\wp(v) = \frac{4g+1}{12g^2}$. The discrete dynamics are recovered in this continuous parametrization via addition of $v$.
\end{theorem}

Remark: $x_n = x(\xi + nv)$.

\begin{proof}
On a level set $I(x,y) = e$ of the invariant, we have the factorization
\begin{align*}
I(x,y) & = xy\left(x+y-\frac{1}{g}\right) + \frac{1}{g}(x+y) -\frac{1}{g^2} \\
 & = \left(x+y-\frac{1}{g}\right)\left(xy+\frac{1}{g}\right) \\
 & = e
\end{align*}
into which can be introduced a new parameter, $t$, so that
\begin{align*}
t & = xy + \frac{1}{g} \\*
\frac{e}{t} & = x+y - \frac{1}{g}.
\end{align*}
Then
\begin{align*}
xy & = t - \frac{1}{g} \\
x + y & = \frac{e}{t} + \frac{1}{g},
\end{align*}
so $x$ and $y$ can be recovered as the roots of the quadratic with coefficients the elementary symmetric polynomials $xy$ and $x+y$:
\bea \label{HJ}
 X^2 - (x+y) X + xy = 0. 
\eea
Thus
\begin{align*}
x, y & = \frac{x+y \mp \sqrt{(x+y)^2 - 4xy}}{2} \\
 & = \frac{\frac{e}{t} + \frac{1}{g} \mp \sqrt{\left(\frac{e}{t} + \frac{1}{g}\right)^2 - 4\left(t - \frac{1}{g}\right)}}{2} \\
 & = \frac{e + \frac{t}{g} \mp \sqrt{\left(e+\frac{t}{g}\right)^2 - 4t^2\left(t-\frac{1}{g}\right)}}{2t}.
\end{align*}
Set
\begin{align*}
w^2 & = \left(e+\frac{t}{g}\right)^2 - 4t^2\left(t-\frac{1}{g}\right) \\
 & = - 4 t^3 + \frac{1}{g^2}(4g+1)t^2 + \frac{2e}{g}t + e^2.
\end{align*}
This curve is put into Weierstrass normal form,
\begin{equation*}
w^2 = -4z^3 + g_2 z - g_3,
\end{equation*}
by the translation $t = z + \alpha$ for $\alpha = \frac{1}{12g^2}(4g+1)$, giving the following for Weierstrass invariants:
\begin{subequations}\label{eq: Weierstrass invariants in alpha}
\begin{align}
g_2 & = 12\alpha^2 + \frac{2e}{g} \\
g_3 & = -8\alpha^3 - \frac{2e}{g}\alpha - e^2.
\end{align}
\end{subequations}
The Weierstrass normal form of an elliptic curve is equivalent to the differential equation 
\[ (\wp')^2 = 4\wp^3 - g_2\wp - g_3 \]
defining the Weierstrass $\wp$-function and its derivative, under the identification $z = -\wp(\xi)$ and $w = \wp'(\xi)$. Make the change of variables
\begin{subequations}
\begin{align}
t & = z + \alpha = -\wp(\xi) + \alpha \\
w & = \wp'(\xi) \\
\alpha & = \wp(v)\label{eq: definition of alpha} \\
e & = \wp'(v)\label{eq: calculation of wp'(v)}
\end{align}
\end{subequations}
where \eqref{eq: definition of alpha} is a definition of $v$ and \eqref{eq: calculation of wp'(v)} follows from the calculation
\begin{align*}
\wp'(v) & = \sqrt{4\wp^3(v) - g_2\wp(v) - g_3} \\
 & = \sqrt{4\alpha^3 - \left(12\alpha^2 + \frac{2e}{g}\right)\alpha - \left( -8\alpha^3 - \frac{2e}{g}\alpha - e^2 \right)} \\
 & = e.
\end{align*}
Under this change of variables, points $(x,y) = \left(x(\xi), y(\xi)\right)$ on the curve $I(x,y) = e$ are given by the parametrization
\begin{equation}
x, y = \frac{\wp'(v) + \frac{\wp(v) - \wp(\xi)}{g} \mp \wp'(\xi)}{2(\wp(v) - \wp(\xi))},
\end{equation}
equivalent to equation \eqref{eq: x, y parametrized through wp}. We separate out the proof of discrete dynamics recovery into Lemma \ref{lem: dPI dynamics recovered}.
\end{proof}

The following corollary is for computational use later on, and is simply a reformulation of the elliptic parametrization in terms of a different function in the Weierstrass function family.

\begin{corollary}\label{cor: parametrization through sigma}
The function $x(\xi)$ of Theorem \ref{thm: elliptic parametrization dPI} can be expressed entirely in terms of the Weierstrass $\sigma$-function:
\begin{equation}
x(\xi) = \frac{1}{2} \left[ \frac{1}{g} - \frac{\sigma(2v)\sigma^4(\xi) - \sigma(2\xi)\sigma^4(v)}{\sigma^2(v)\sigma^2(\xi)\sigma(\xi-v)\sigma(\xi+v)} \right].
\end{equation}
\end{corollary}
\begin{proof}
We make use of the following two identities (\cite{Mar77}, equation III.5.52 and exercise III.5.16) which relate the Weierstrass $\sigma$-function to the $\wp$-function:
\begin{equation}\label{eq: wp and sigma identity}
\wp(z) - \wp(Z) = - \frac{\sigma(z-Z)\sigma(z+Z)}{\sigma^2(z)\sigma^2(Z)}
\end{equation}
if $Z$ is not a period of $\wp(z)$ and
\begin{equation}
\wp'(z) = - \frac{\sigma(2z)}{\sigma^4(z)}.
\end{equation}
Then
\begin{align*}
x(\xi) & = \frac{1}{2} \left(\frac{1}{g} + \frac{\wp'(v) - \wp'(\xi)}{\wp(v) - \wp(\xi)} \right) \\
 & = \frac{1}{2} \left[ \frac{1}{g} - \left( \frac{\sigma(2v)}{\sigma^4(v)} - \frac{\sigma(2\xi)}{\sigma^4(\xi)} \right)\left( - \frac{\sigma^2(v) \sigma^2(\xi)}{\sigma(v-\xi)\sigma(v+\xi)} \right) \right] \\
 & = \frac{1}{2} \left[ \frac{1}{g} + \left( \frac{\sigma(2v)\sigma^4(\xi) - \sigma(2\xi)\sigma^4(v)}{\sigma^4(v)\sigma^4(\xi)} \right)\left( \frac{\sigma^2(v) \sigma^2(\xi)}{\sigma(v-\xi)\sigma(v+\xi)} \right) \right] \\
 & = \frac{1}{2} \left[ \frac{1}{g} + \left( \frac{1}{\sigma^2(v)\sigma^2(\xi)} \right)\left( \frac{\sigma(2v)\sigma^4(\xi) - \sigma(2\xi)\sigma^4(v)}{\sigma(v-\xi)\sigma(v+\xi)} \right) \right] \\
 & = \frac{1}{2} \left[ \frac{1}{g} - \frac{\sigma(2v)\sigma^4(\xi) - \sigma(2\xi)\sigma^4(v)}{\sigma^2(v)\sigma^2(\xi)\sigma(\xi-v)\sigma(\xi+v)} \right], \\
\end{align*}
the last line since $\sigma$ is an odd function.
\end{proof}

\begin{lemma}\label{lem: dPI dynamics recovered}
Under the parametrization of Theorem \ref{thm: elliptic parametrization dPI}, the discrete dynamics of the system are recovered via addition of $v$. That is,
\begin{subequations}\label{eq: x, y dynamics recovered}
\begin{align}
x(\xi + v) & = \frac{1}{g} - \frac{1}{gx(\xi)} - x(\xi) - y(\xi)\label{eq: x dynamics recovered} \\
y(\xi + v) & = x(\xi)\label{eq: y dynamics recovered}
\end{align}
\end{subequations}
where $v$ is such that $\wp(v) = \frac{4g+1}{12g^2}$.
\end{lemma}

\begin{proof}
The lemma rests solely on the use of certain properties of the Weierstrass functions, including the Weierstrass $\zeta$-function given by $\zeta'(u) = - \wp(u)$: the fact that $\zeta$ is an odd function, 
the identity (\cite{Mar77} equation (5.59))
\begin{equation}\label{eq: wp and zeta identity from Ahlfors}
\frac{\wp'(u)}{\wp(u) - \wp(v)} = \zeta(u+v) + \zeta(u-v) - 2\zeta(u),
\end{equation}
and the addition law for $\wp$ (\cite{Mar77} Theorem 5.15)
\begin{equation}\label{eq: wp addition law}
\wp(u+v) + \wp(u) + \wp(v) = \frac{1}{4} \left( \frac{\wp'(u) - \wp'(v)}{\wp(u) - \wp(v)} \right)^2.
\end{equation}
We first calculate that
\begin{equation}\label{eq: lemma1 eq1}
x(\xi), y(\xi) = \frac{1}{2g} + \zeta(\pm \xi + v) \mp \zeta(\xi) - \zeta(v),
\end{equation}
starting from \eqref{eq: x, y parametrized through wp}, expanding, and using \eqref{eq: wp and zeta identity from Ahlfors}:
\begin{align*}
x(\xi), y(\xi) & = \frac{1}{2g} + \frac{1}{2} \left( \zeta(v+\xi) + \zeta(v-\xi) - 2\zeta(v)\right) \pm \frac{1}{2}\left( \zeta(\xi+v) + \zeta(\xi-v) - 2\zeta(\xi)\right) \\
 & = \frac{1}{2g} + \frac{1}{2} \left( \zeta(\xi+v) \pm \zeta(\xi+v)\right) + \frac{1}{2}\left( \zeta(-\xi+v) \mp \zeta(-\xi+v) \right) - \left(\zeta(v) \pm \zeta(\xi)\right) \\
 & = \frac{1}{2g} + \zeta(\pm\xi+v) \mp\zeta(\xi) - \zeta(v).
\end{align*}

Demonstration of the result \eqref{eq: y dynamics recovered} is now direct via evaluation at $\xi+v$:
\begin{align*}
y(\xi+v) & = \frac{1}{2g} + \zeta(-(\xi+v)+v) + \zeta(\xi+v) - \zeta(v) \\
 & = \frac{1}{2g} + \zeta(-\xi) + \zeta(\xi+v) - \zeta(v) \\
 & = \frac{1}{2g} - \zeta(\xi) + \zeta(\xi+v) - \zeta(v) \\
 & = x(\xi).
\end{align*}

Demonstration of the other part of the result proceeds in a few more steps. First compute the product $x(\xi)y(\xi)$ using the forms in \eqref{eq: x, y parametrized through wp}, expanding along the way $\wp'(v)^2 - \wp'(\xi)^2 = -4(2\wp(v) + \wp(\xi))(\wp(v)-\wp(\xi))^2 - \frac{2\wp'(v)}{g}(\wp(v) - \wp(\xi))$ using the differential equation; thus
\begin{equation}\label{eq: lemma eq1}
x(\xi)y(\xi) = \frac{1}{4g^2} - \left( 2\wp(v) + \wp(\xi) \right).
\end{equation}
Starting with the left-hand side of the following equality, expanding $\wp(\xi+v)$ using the addition formula \eqref{eq: wp addition law}, and then making a substitution based on \eqref{eq: x parametrized through wp} give that
\begin{equation}\label{eq: lemma eq2}
\wp(v) - \wp(\xi+v) = 2\wp(v) + \wp(\xi) - \left( x(\xi) - \frac{1}{2g} \right)^2,
\end{equation}
which is inserted into 
\begin{equation}\label{eq: lemma eq3}
x(\xi+v) + x(\xi) = \frac{1}{g} + \frac{\wp'(v)}{\wp(v) - \wp(\xi+v)},
\end{equation}
the result of summing $x$ and $y$ from \eqref{eq: x, y parametrized through wp} evaluated at $\xi + v$, along with \eqref{eq: y dynamics recovered}. Inserting \eqref{eq: lemma eq1} into \eqref{eq: lemma eq2} yields $\wp(v) - \wp(\xi+v) = -x(\xi)\left(x(\xi) + y(\xi) - \frac{1}{g} \right)$, and recall that 
\[ x + y - \frac{1}{g} = \frac{\wp'(v)}{\wp(v) - \wp(\xi)} \hspace{.15in} \mbox{and} \hspace{.15in} xy + \frac{1}{g} = \wp(v) - \wp(\xi). \]
Putting all of these together into \eqref{eq: lemma eq3} the remaining result \eqref{eq: x dynamics recovered} is concluded:
\begin{align*}
x(\xi+v) & = -x(\xi) + \frac{1}{g} + \frac{\wp'(v)}{-x(\xi)\left( x(\xi) + y(\xi) - \frac{1}{g} \right)} \\
 & = -x(\xi) + \frac{1}{g} + \frac{\wp'(v)}{-x(\xi) \frac{\wp'(v)}{\wp(v) - \wp(\xi)}} \\
 & = -x(\xi) + \frac{1}{g} - \frac{x(\xi)y(\xi) + \frac{1}{g}}{x(\xi)} \\
 & = \frac{1}{g} -\frac{1}{gx(\xi)} - x(\xi) - y(\xi).
\end{align*}
\end{proof}

\subsection{Separatrix Degeneration to the Combinatorial Orbit}
In the previous section, elliptic parametrizations were given for the $c = 1$ system, relying on the invariant $I$. Orbits of the system lie on level sets $I(x,y) =e$, and we have so far treated the generic case in which $I(x,y) = e$ is a nonsingular curve. It is in this generic setting that results of the previous section hold. In this section, we study certain cases in which $I(x,y) = e$ is singular, and we extend the parametrizations to certain such curves, in particular, the ones of combinatorial interest. Recall that Table \ref{table:DegenerateEnergies} contains all the energy values resulting in singular level sets for the two invariants. We will show at the appropriate time that the energy $e_1$ in the $c=1$ case is the one corresponding to geodesic distance for quadrangulations. However, this can be observed in Figure \ref{fig: separatrices Rn} as it is the $e_1$ separatrix which contains the hyperbolic fixed point, with coordinates given by $(x^*_-, x^*_-)$ and in the combinatorial problems the limit $x^* = \lim_{n\rightarrow\infty} x_n$ of the generating functions is the analytic function given by the root $x^*_-$.

\begin{theorem}\label{thm: main theorem dPI}
The discrete integrable recurrence
\begin{equation}
x_n = 1 + gx_n(x_{n-1} + x_n + x_{n+1})
\end{equation}
with invariant
\begin{equation}
I(x_{n}, x_{n-1}) = x_nx_{n-1}\left(x_n+x_{n-1}-\frac{1}{g} \right) + \frac{1}{g}(x_n+x_{n-1}) - \frac{1}{g^2}
\end{equation}
and limit $\lim_{n\rightarrow\infty} x_n = x_-^*$ satisfying $x_-^* = 1 + 3gx_-^{*2}$ possesses a solution
\begin{equation}
x_n = x(\xi + n\nu) = x_-^* \cdot \frac{\left( 1 - e^{2\xi\sqrt{-3r_2}} \chi^{n-1} \right)\left( 1 - e^{2\xi\sqrt{-3r_2}} \chi^{n+2} \right)}{\left( 1 - e^{2\xi\sqrt{-3r_2}} \chi^{n} \right)\left( 1 - e^{2\xi\sqrt{-3r_2}} \chi^{n+1} \right)},
\end{equation} 
where $\nu$ is such that $\sinh(\nu\sqrt{-3r_2}) = \sqrt{\frac{-3r_2}{\alpha+r_2}}$ and $\chi = \exp\{-2\nu\sqrt{-3r_2}\}$.

In particular, the combinatorial generating functions $R_n(g) = x_n$ and $R(g) = x_-^*$ specialized by the additional boundary condition $R_{-1} = 0$ have the closed form
\begin{equation}
R_n = R \frac{(1-\chi^{n+1})(1-\chi^{n+4})}{(1-\chi^{n+2})(1-\chi^{n+3})}
\end{equation}
for $\chi$ satisfying $\chi + \frac{1}{\chi} = \frac{1-4gR}{gR}$, by specializing $\xi = -2\nu$.
\end{theorem}

An orbit $\{(x_n, y_n)\}$ satisfying $\lim_{n\rightarrow\infty} x_n = x^{*}_-$ limits to the hyperbolic fixed point $(x^*_-, x^*_-)$, and thus the energy of the invariant level set on which this orbit lies can be calculated by $e = I(x^*_-, x^*_-) = I(R, R)$. 
\begin{equation}\label{eq: I(R,R) = e1}
\begin{aligned}
I(R, R) & = \left( 2R - \frac{1}{g} \right) \left( R^2 + \frac{1}{g} \right) \\
 & = \left[ 2\left(\frac{1 - \sqrt{1-12g}}{6g}\right) - \frac{1}{g} \right] \left[ \left( \frac{1 - \sqrt{1-12g}}{6g}\right)^2 + \frac{1}{g} \right] \\
 & = \left( \frac{- 2 - \sqrt{1-12g}}{3g} \right) \left( \frac{1 + 12g - \sqrt{1-12g}}{18g^2} \right) \\
 & = -\frac{1 + 36g + (12g-1)\sqrt{1-12g}}{54g^3} \\
 & = e_1
\end{aligned}
\end{equation}
In the following Proposition, the elliptic parametrization is extended to orbits on the level set $I(x,y) = e_1$, and the proof of the Theorem is deferred until after this degeneration of the parametrization is established.

\begin{proposition}\label{prop: PI degeneration}
The elliptic parametrization of Theorem \ref{thm: elliptic parametrization dPI} and Corollary \ref{cor: parametrization through sigma} extends to a parametrization of the singular curve 
\[ I(x,y) = e_1 = -\frac{1 + 36g + (12g-1)\sqrt{1-12g}}{54g^3} \]
as
\begin{equation}\label{eq: x via degenerate sigma}
x(\xi+n\nu) = \frac{1}{2} \left[ \frac{1}{g} - \frac{\sigma_d(2\nu)\sigma_d^4(\xi+n\nu) - \sigma_d(2(\xi+n\nu))\sigma_d^4(\nu)}{\sigma_d^2(\nu)\sigma_d^2(\xi+n\nu)\sigma_d(\xi+(n-1)\nu)\sigma_d(\xi+(n+1)\nu)} \right],
\end{equation}
where the function $\sigma_d$ is a degeneration of the Weierstrass $\sigma$-function, given explicitly by the formula
\begin{equation}
\sigma_d(\xi) = \exp \left\{\frac{r_2}{2} \xi^2 \right\} \frac{\sinh(\xi\sqrt{-3r_2})}{\sqrt{-3r_2}},
\end{equation}
where 
\[ r_2 = \frac{-4(1-12g + 2\sqrt{1-12g})}{12^2g^2}. \]
The related functions in this degeneration of the family of Weierstrass functions are given explicitly by
\begin{align*}
\wp_d(\xi) & = - r_2 -3r_2\csch^2(\xi\sqrt{-3r_2}) \\
\wp_d'(\xi) & = -2(-3r_2)^{3/2} \csch^2(\xi\sqrt{-3r_2})\coth(\xi\sqrt{-3r_2}) \\
\zeta_d(\xi) & = r_2\xi + \sqrt{-3r_2} \coth(\xi\sqrt{-3r_2}).
\end{align*}
\end{proposition}

\begin{proof}
With the energy $e$ specialized to $e_1(g)$, the two variables $g$ and $e$ in the system are tied together, and we pursue further calculations after expressing everything in sight in terms of a new fundamental variable, $\gamma$, which suppresses the square root appearing in $e_1$:
\begin{align*}
\gamma & = \sqrt{1-12g} \\
g & = \frac{1-\gamma^2}{12} \\
e_1 & = \frac{-32(\gamma-1)(\gamma+2)^2}{(\gamma+1)^3(\gamma-1)^3} \\
g_2 & = 12\frac{4^2 \gamma^2 (\gamma+2)^2}{(\gamma+1)^4 (\gamma-1)^4} \\
g_3 & = -8\frac{4^3 \gamma^3 (\gamma+2)^3}{(\gamma+1)^6 (\gamma-1)^6} \\
\alpha & = \frac{-4(\gamma+2)(\gamma-2)}{(\gamma+1)^2(\gamma-1)^2} \\
R & = \frac{2}{\gamma+1}.
\end{align*}

On the singular curve $I(x,y) = e_1$, at least two roots of $-4z^3 + g_2z - g_3$ coalesce. We therefore seek a factorization
\[ w^2 = -4z^3 + g_2z - g_3 = -4(z-r_1)(z-r_2)^2, \]
where $r_1$ and $r_2$ are not assumed to be distinct. Equating coefficients in the above equation gives that $r_1 + 2r_2 = 0$, $g_2 = - 4(2r_1r_2 + r_2^2)$, and $g_3 = - 4r_1r_2^2$. Therefore, $r_1 = - 2r_2$, $g_2 = 12r_2^2$, and $g_3 = 8r_2^3$. Thus,
\begin{align*}
r_2 & = \frac{-4\gamma(\gamma+2)}{(\gamma+1)^2(\gamma-1)^2} = \frac{-4(1-12g + 2\sqrt{1-12g})}{12^2g^2} \\
r_1 & = \frac{8\gamma(\gamma+2)}{(\gamma+1)^2(\gamma-1)^2} = \frac{8(1-12g + 2\sqrt{1-12g})}{12^2g^2}.
\end{align*}

The identification $z = - \wp_d(\xi)$ and $w = \wp_d'(\xi)$ is made just as before, to a function $\wp_d(\xi)$ and its derivative satisfying the same differential equation as the Weierstrass $\wp$-function,
\[ (\wp_d')^2 = 4\wp_d^3 - g_2 \wp_d - g_3 = -4(-\wp_d(\xi)-r_1)(-\wp_d(\xi)-r_2)^2. \]
The $\wp$-function is given as the inverse of an elliptic integral (\cite{Mar77}, Eq. (5.28))
\[ \xi = \int_{\infty}^{\wp(\xi)} \frac{dt}{\sqrt{4t^3 - g_2t - g_3}}, \]
and the related, degenerate, function $\wp_d(\xi)$ is calculated directly by the same integral:
\[ \xi = \int_{\infty}^{\wp_d(\xi)} \frac{dt}{-2(t+r_2)\sqrt{t+r_1}}. \]
The result of the integration is
\[ \xi = \left. -\frac{1}{\sqrt{r_1-r_2}} \tanh^{-1} \left( \sqrt{\frac{t + r_1}{r_1-r_2}} \right) \right|_{\infty}^{\wp_d(\xi)}, \]
and inverting the equation gives an explicit formula for this degeneration of the Weierstrass elliptic function:
\begin{equation}
\wp_d(\xi) = - r_2 -3r_2\csch^2(\xi\sqrt{-3r_2}).
\end{equation}
The derivative (with respect to $\xi$) is given by
\begin{equation}
\wp_d'(\xi) = -2(-3r_2)^{3/2} \csch^2(\xi\sqrt{-3r_2})\coth(\xi\sqrt{-3r_2}).
\end{equation}
The Weierstrass $\zeta$-function, an anti-derivative of $\wp$, is defined by (\cite{Mar77} 5.40)
\[ \frac{d}{d\xi} \zeta(\xi) = -\wp(\xi) \hspace{.15in} \mbox{and} \hspace{.15in} \lim_{\xi\rightarrow 0} \left[ \zeta(\xi) - \frac{1}{\xi} \right] = 0. \]
Integrating $\wp_d(\xi)$ and pinning down the integration constant by the limit condition using the first few terms of a Taylor expansion for $\coth(y)$ about $y=0$, we have the degeneration
\begin{equation}\label{eq: zeta degenerate}
\zeta_d(\xi) = r_2\xi + \sqrt{-3r_2} \coth(\xi\sqrt{-3r_2}).
\end{equation}
The Weierstrass $\sigma$-function is a logarithmic antiderivative of $\zeta$ (\cite{Mar77} 5.44) given by
\[ \frac{d}{d\xi} \log \sigma(\xi) = \frac{\sigma'(\xi)}{\sigma(\xi)} = \zeta(\xi) \hspace{.15in} \mbox{and} \hspace{.15in} \lim_{\xi\rightarrow 0} \frac{\sigma(\xi)}{\xi} = 1, \]
which can be calculated by the integral (\cite{Mar77} 5.45)
\[ \sigma(\xi) = \xi\exp\left\{ \int_0^\xi \left[ \zeta(t) - \frac{1}{t} \right] dt \right\}. \]
Thus by integrating the degenerate $\zeta_d$,
\[ \sigma_d(\xi) = \xi \exp \left\{ C + \frac{r_2}{2} \xi^2 + \log \left( \frac{\sinh(\xi\sqrt{-3r_2})}{\xi} \right) \right\} \]
with integration constant $C = \log\left( \frac{1}{\sqrt{-3r_2}} \right)$, and the degenerate $\sigma_d(\xi)$ is given by
\begin{equation}\label{eq: sigma degenerate}
\sigma_d(\xi) = \exp \left\{\frac{r_2}{2} \xi^2 \right\} \frac{\sinh(\xi\sqrt{-3r_2})}{\sqrt{-3r_2}}.
\end{equation}

Theorem \ref{thm: elliptic parametrization dPI} now goes through with all of the Weierstrass functions replaced by their appropriate degenerations. The proofs of Corollary \ref{cor: parametrization through sigma} and Lemma \ref{lem: dPI dynamics recovered} rely on three identities of the Weierstrass functions: equations \eqref{eq: wp and sigma identity}, \eqref{eq: wp and zeta identity from Ahlfors}, \eqref{eq: wp addition law}, and the oddness of $\zeta$, which can all be verified to hold in the degenerate case as well. The proofs are just algebra along with hyperbolic trigonometric identities; one could, for example, use the Pythagorean Theorem to write all hyperbolic trigonometric functions in terms of the hyperbolic cotangent, and then use simply the addition laws for the hyperbolic cotangent. Once these identities are established, all of the general parametrization and dynamics machinery extend directly to the $I(x,y) = e_1$ separatrix and Equation \eqref{eq: x via degenerate sigma} is immediate.
\end{proof}

\begin{proof}[Proof of Theorem \ref{thm: main theorem dPI}]
The results of Proposition \ref{prop: PI degeneration} give the degeneration of the elliptic parametrization appropriate for the Theorem, as established by the calculation of the energy $I(x^*_-, x^*_-) = e_1$ in equation \eqref{eq: I(R,R) = e1}. Inserting the explicit formula for $\sigma_d$ given in equation \eqref{eq: sigma degenerate} into the parametrization of equation \eqref{eq: x via degenerate sigma}, observe that all of the exponentials cancel out and an overall factor of $\sqrt{-3r_2}$ rests in the numerator of the large fraction. Then, scaling the variables temporarily to suppress factors of $\sqrt{-3r_2}$ appearing in every argument, via
\begin{align*}
\phi & = \xi\sqrt{-3r_2} \\
\eta & = \nu\sqrt{-3r_2},
\end{align*}
and simplifying using the double-angle formula for the hyperbolic sine, we have
\begin{align*}
x & \left( \xi + nv \right) \\
 & = \frac{1}{2g} - \frac{\sqrt{-3r_2}}{2} \left( \frac{\sinh(2\eta)\sinh^3(\phi+n\eta)}{\sinh^2(\eta)\sinh(\phi+n\eta)\sinh(\phi+(n-1)\eta)\sinh(\phi+(n+1)\eta)} \right.\\
 & \phantom{ = a bunch of room } \left. + \frac{\sinh(2(\phi+n\eta))\sinh^3(\eta)}{\sinh(\eta)\sinh^2(\phi+n\eta)\sinh(\phi+(n-1)\eta)\sinh(\phi+(n+1)\eta)} \right) \\
 & = \frac{1}{2g} - \sqrt{-3r_2} \left( \frac{\cosh(\eta)\sinh^3(\phi+n\eta)}{\sinh(\eta)\sinh(\phi+n\eta)\sinh(\phi+(n-1)\eta)\sinh(\phi+(n+1)\eta)} \right. \\
 & \phantom{ = a bunch of room } \left. + \frac{\cosh(\phi+n\eta)\sinh^3(\eta)}{\sinh(\eta)\sinh(\phi+n\eta)\sinh(\phi+(n-1)\eta)\sinh(\phi+(n+1)\eta)} \right) \\
 & = \frac{1}{2g} - \sqrt{-3r_2} \left( \coth(\eta)\frac{\sinh^3(\phi+n\eta)}{\sinh(\phi+n\eta)\sinh(\phi+(n-1)\eta)\sinh(\phi+(n+1)\eta)} \right. \\
 & \phantom{ = a bunch of room } \left. + \coth(\phi+n\eta)\frac{\sinh^3(\eta)}{\sinh(\eta)\sinh(\phi+(n-1)\eta)\sinh(\phi+(n+1)\eta)} \right).
\end{align*}
By definition of each hyperbolic trigonometric function in terms of exponentials, and defining
\[ \chi = e^{-2\eta} = e^{-2v\sqrt{-3r_2}}, \]
then
\begin{align*}
x\left( \xi + nv \right) & = \frac{1}{2g} - \sqrt{-3r_2} \left[ \left( \frac{\chi+1}{\chi-1} \right) \cdot \frac{(1-e^{2\phi}\chi^n)^2}{(1-e^{2\phi}\chi^{n-1})(1-e^{2\phi}\chi^{n+1})} \right. \\
 & \left. \phantom{ = stuff } \hspace{2cm} + \left( \frac{1+e^{2\phi}\chi^n}{1-e^{2\phi}\chi^n} \right) \cdot \frac{e^{2\phi}\chi^n(1-\chi)^2}{\chi(1-e^{2\phi}\chi^{n-1})(1-e^{2\phi}\chi^{n+1})} \right].
\end{align*}

We Taylor expand for large $n$ (note that $\chi^n$ is small). First, the various pieces:
\begin{align*}
\frac{(1-e^{2\phi}\chi^n)^2}{(1-e^{2\phi}\chi^{n-1})(1-e^{2\phi}\chi^{n+1})} & = (1-e^{2\phi}\chi^n)^2 \sum_{j=0}^{\infty} \left(e^{2\phi}\chi^{n-1}\right)^j \sum_{k=0}^{\infty} \left(e^{2\phi}\chi^{n+1}\right)^k \\
 & = 1 + \sum_{k=1}^{\infty} e^{2\phi k} \chi^{nk} \left( \sum_{j=0}^{k-1} \left( \chi^{k-2j} - 2\chi^{k-1-2j} + \chi^{k-2-2j} \right) \right) \\
\left( \frac{1+e^{2\phi}\chi^n}{1-e^{2\phi}\chi^n} \right) & = 1 + \sum_{k=1}^{\infty} 2e^{2\phi k} \chi^{nk} \\
\frac{e^{2\phi}\chi^n(1-x)^2}{\chi(1-e^{2\phi}\chi^{n-1})(1-e^{2\phi}\chi^{n+1})} & = e^{2\phi}\chi^{n}\left(\chi-2+\frac{1}{\chi}\right) \sum_{k=0}^{\infty} e^{2\phi k}\chi^{nk} \left( \sum_{j=0}^{k} \chi^{k-2j} \right) \\
 & = \left(\chi-2+\frac{1}{\chi}\right) \sum_{k=1}^{\infty} e^{2\phi k}\chi^{nk} \left( \sum_{j=0}^{k-1} \chi^{k-1-2j} \right) \\
 & = \sum_{k=1}^{\infty} e^{2\phi k}\chi^{nk} \left( \sum_{j=0}^{k-1} \chi^{k-2j} - 2\chi^{k-1-2j} + \chi^{k-2-2j} \right). \\
\end{align*}
The product of the last two simplifies as
\begin{align*}
\left( \frac{1+e^{2\phi}\chi^n}{1-e^{2\phi}\chi^n} \right)& \cdot \frac{e^{2\phi}\chi^n(1-x)^2}{\chi(1-e^{2\phi}\chi^{n-1})(1-e^{2\phi}\chi^{n+1})} \\
& = \sum_{k=1}^{\infty} e^{2\phi k} \chi^{nk} \sum_{l=0}^{k-1} \left( \chi^{k-2l} - 2\chi^{k-1-2l} + \chi^{k-2-2l} \right) \\
& \quad+ \sum_{k=2}^{\infty} \sum_{m=1}^{k-1} e^{2\phi m}\chi^{nm} \sum_{j=0}^{m-1} \left( \chi^{m-2j} - 2\chi^{m-1-2j} + \chi^{m-2-2j} \right) 2e^{2\phi(k-m)} \chi^{n(k-m)} \\
& = \sum_{k=1}^{\infty} e^{2\phi k} \chi^{nk} \sum_{l=0}^{k-1} \left( \chi^{k-2l} - 2\chi^{k-1-2l} + \chi^{k-2-2l} \right) \\
& \phantom{ =  } + \sum_{k=2}^{\infty} e^{2\phi k} \chi^{nk} \sum_{m=1}^{k-1} \sum_{j=0}^{m-1} 2\left( \chi^{m-2j} - 2\chi^{m-1-2j} + \chi^{m-2-2j} \right) \\
& = \sum_{k=1}^{\infty} e^{2\phi k} \chi^{n} \left( \chi^k - 2 + \chi^{-k} \right)
\end{align*}
since for $k\ge 2$ the following sum holds by telescoping the series:
\begin{align*} 
  &\sum_{l=0}^{k-1} \left( \chi^{k-2l} - 2\chi^{k-1-2l} + \chi^{k-2-2l} \right) + \sum_{m=1}^{k-1} \sum_{j=0}^{m-1} 2\left( \chi^{m-2j} - 2\chi^{m-1-2j} + \chi^{m-2-2j} \right)\\ 
  &\quad\quad= \chi^k - 2 + \chi^{-k}.
\end{align*}
Now putting all these pieces together, we have
\begin{align*}
&x(\xi + nv) \\
& = \frac{1}{2g} - \sqrt{-3r_2} \left[ \frac{\chi+1}{\chi-1} + (\chi+1)\sum_{k=1}^{\infty} e^{2\phi k} \chi^{nk} \sum_{j=0}^{k-1} \frac{ \chi^{k-2j} - 2\chi^{k-1-2j} + \chi^{k-2-2j}}{\chi-1} \right. \\
 & \phantom{ = and more space } \left. +\sum_{k=1}^{\infty} e^{2\phi k} \chi^{n} \left( \chi^k - 2 + \chi^{-k} \right) \right] \\
 & = \frac{1}{2g} - \sqrt{-3r_2} \left[ \frac{\chi+1}{\chi-1} + (\chi+1) \sum_{k=1}^{\infty} e^{2\phi k} \chi^{nk} \sum_{j=0}^{k-1} \left( \chi^{k-2j-1} - \chi^{k-2-2j} \right) \right. \\
 & \phantom{ = and more space } \left. 
+ \sum_{k=1}^{\infty} e^{2\phi k} \chi^{n} \left( \chi^k - 2 + \chi^{-k} \right) \right] \\
 & = \frac{1}{2g} - \sqrt{-3r_2} \left[ \frac{\chi+1}{\chi-1} + \sum_{k=1}^{\infty} e^{2\phi k} \chi^{nk} \sum_{j=0}^{k-1} \left( \chi^{k-2j} - \chi^{k-2-2j} \right) \right. \\
 & \phantom{ = and more space } \left. + \sum_{k=1}^{\infty} e^{2\phi k} \chi^{n} \left( \chi^k - 2 + \chi^{-k} \right) \right] \\
 & = \frac{1}{2g} - \sqrt{-3r_2} \left[ \frac{\chi+1}{\chi-1} + \sum_{k=1}^{\infty} e^{2\phi k} \chi^{nk} \left( \chi^k - \chi^{-k} \right) + \sum_{k=1}^{\infty} e^{2\phi k} \chi^{n} \left( \chi^k - 2 + \chi^{-k} \right) \right] \\
 & = \frac{1}{2g} - \sqrt{-3r_2} \left[ \frac{\chi+1}{\chi-1} + \sum_{k=1}^{\infty} e^{2\phi k} \chi^{nk} 2\left(\chi^k - 1\right) \right].
\end{align*}
The constants out front are simplified by first observing that
\[ \sqrt{-3r_2} = R \cdot \frac{\chi^2-1}{2\chi}, \]
by writing the fraction in $\chi$ as a product of the hyperbolic sine and cosine and squaring both sides of the equation, so that everything can be reduced to its expression in terms of $\gamma$. Next, we claim that 
\[ \chi + \frac{1}{\chi} = \frac{1-4gR}{gR}, \]
which is easily proved by expressing $\sinh\left(v\sqrt{-3r_2}\right)$ in terms of $\chi$ and then reducing everything in sight to its expression as a function of $\gamma$. Using this second identity, it is also easily proved that
\[ \frac{1}{2gR} - \frac{(\chi+1)^2}{2\chi} = 1, \]
from which the expression for $x(\xi + nv)$ is finally re-summable:
\begin{align*}
x&(\xi+nv) \\
& = R \left[ 1 - \frac{\chi^2-1}{\chi} \sum_{k=1}^{\infty} e^{2\phi k} \left( \chi^{(n+1)k} - \chi^{nk} \right) \right] \\
 & = R \left[ 1 - \frac{\chi^2-1}{\chi} \left( \frac{1}{1-e^{2\phi}\chi^{n+1}} - \frac{1}{1-e^{2\phi}\chi^n} \right) \right] \\
 & = R \cdot \frac{\chi(1-e^{2\phi}\chi^{n+1})(1-e^{2\phi}\chi^{n}) - (\chi^2-1)(1-e^{2\phi}\chi^{n}) + (\chi^2-1)(1-e^{2\phi}\chi^{n+1})}{\chi(1-e^{2\phi}\chi^{n})(1-e^{2\phi}\chi^{n+1})} \\
 & = R \cdot \frac{1-e^{2\phi}\chi^n - e^{2\phi}\chi^{n+1} + e^{4\phi}\chi^{2n+1} + e^{2\phi}\chi^{n+1} - e^{2\phi}\chi^{n-1} - e^{2\phi}\chi^{n+2} + e^{2\phi}\chi^{n}}{(1-e^{2\phi}\chi^{n})(1-e^{2\phi}\chi^{n+1})} \\
 & = R \cdot \frac{1 + e^{4\phi}\chi^{2n+1} - e^{2\phi}\chi^{n-1} - e^{2\phi}\chi^{n+2}}{(1-e^{2\phi}\chi^{n})(1-e^{2\phi}\chi^{n+1})}  \\
 & = R \cdot \frac{(1-e^{2\phi}\chi^{n-1})(1-e^{2\phi}\chi^{n+2})}{(1-e^{2\phi}\chi^{n})(1-e^{2\phi}\chi^{n+1})}.
\end{align*}
\end{proof}


\section{Elliptic Parametrization, c=0} \label{e-param0}
\subsection{The Generic Orbit}

The discrete system (\ref{eq: discrete dynamics general c}) with $c=0$  has the invariant, 
(\ref{eq: integral of motion c=0}), 
\begin{eqnarray} \label{T-invariant}
e &=& (1+g)xy - y(1 + g x^2) - x(1 + g y^2) + g^2 x^2 y^2 + \frac1g,
\end{eqnarray}
which is quartic, or in fact bi-quadratic. This does not lend itself naturally to a canonical Weierstrass parametrization as in the case of $c=1$ where the invariant was cubic. 
Therefore we will be taking a more geometric approach to the analysis of the $c=0$ case.

For notational convenience we will sometimes pass to the translated energy variable
\beann
E &=& e - \frac1g.
\eeann

To understand the elliptic parametrization of a generic level set of \eqref{eq: integral of motion c=0} we pass to a projective completion of the $x-y$ plane as $\mathbb{P}^1 \times \mathbb{P}^1$ in which the dynamical system becomes
\bea \label{TMap}
\left[\bar{x}_0 : \bar{x}_1\right] &=& \left[g x_1 y_0 : x_1 y_0 - x_0 y_0 - g x_1 y_1 \right] \\ \nonumber
\left[\bar{y}_0 : \bar{y}_1 \right] &=& [{x_0}: {x_1} ]
\eea
This stems from its more usual form in affine coordinates:
\bea \label{projDS1}
\bar{x} & = & \frac{x - 1}{g x} - y \\ \label{projDS2}
\bar{y} & = & x
\eea
This mapping is reversible with inverse given by
\beann
x &=& \bar{y} \\
y &=& \frac{\bar{y} - 1}{g \bar{y}} - \bar{x}
\eeann

The invariant (\ref{T-invariant}) for this system can also be expressed in homogeneous coordinates:
\begin{eqnarray} \label{TInv}
&& I([\bar{x_0}: \bar{x_1}], [\bar{y_0}: \bar{y_1}]) \\ \nonumber
&=& 
(1+g) x_0 x_1 y_0 y_1 -  y_0 y_1 (x_0^2 + g x_1^2 )  - x_0 x_1 (y_0^2 + g y_1^2) + g^2 x_1^2 y_1^2 - E x_0^2 y_0^2 = 0 
\end{eqnarray}
One may check directly that the level sets of this invariant are generically smooth and, in particular, they are all smooth along the lines at infinity: $x_0 = 0$ and $y_0 = 0$. 

Now we may consider the Segre embedding of $\mathbb{P}^1 \times \mathbb{P}^1$ into 
$\mathbb{P}^3$ given by
\bea \label{Segre1}
 ([\bar{x_0}: \bar{x_1}], [\bar{y_0}: \bar{y_1}]) &\to&  \left[Z_0 : Z_1 : Z_2 : Z_3\right] \\
\label{Segre2}
&=& \left[x_0 y_0 : x_0 y_1 : x_1 y_0 : x_1 y_1\right].
\eea
It is immediate from this representation that the Segre map embeds $\mathbb{P}^1 \times \mathbb{P}^1$ as a quadratic surface in $\mathbb{P}^3$ whose equation is given by
\beann
F(\left[Z_0 : Z_1 : Z_2 : Z_3\right]) &=& Z_0 Z_3 - Z_1 Z_2 = 0.
\eeann
It is also straightforward to check that, for each value of $E$, the invariant curve (\ref{TInv}) corresponds to the intersection of the surface $\{F = 0\}$ with another quadric surface in $\mathbb{P}^3$ explicitly given by
\beann
G_E(\left[Z_0 : Z_1 : Z_2 : Z_3\right]) &=& \\
 (1+g)  Z_0 Z_3 &-& \left( Z_0 Z_1 + g Z_2 Z_3\right) - \left(  Z_0 Z_2 + g Z_1 Z_3\right)  + g^2 Z_3^2 - E Z_0^2 = 0.
\eeann
It is well known that the intersection of two smooth quadrics in $\mathbb{P}^3$ is, generically, a smooth elliptic curve: a smooth space curve of degree 4, which we will sometimes refer to as a space quartic.

Returning to energy parameter $e$, the invariant is seen to factor as 
\beann
G_e(\left[Z_0 : Z_1 : Z_2 : Z_3\right]) &=& \left( gZ_3 + Z_0\right) \left( g Z_3 + \frac1g Z_0 - Z_1 - Z_2\right) - e Z_0^2.
\eeann
We work with two additional models of the level set associated to this fundamnetal model of the elliptic space curve. These are based on two projections:
\beann
\pi_1: \mathbb{P}^3 &\to& \mathbb{P}^2\\
\left[Z_0 : Z_1 : Z_2 : Z_3\right] &\to & \left[Z_0 : Z_1 : Z_2 \right];
\eeann

\beann
\pi_2: \mathbb{P}^3 &\to & \mathbb{P}^1\\
\left[Z_0 : Z_1 : Z_2 : Z_3\right] &\to & \left[Z_0 : Z_3\right].
\eeann

To understand the geometric meaning of these  projections, set $t =  \left( gZ_3 + Z_0\right)$, and also note that from (\ref{Segre1} - \ref{Segre2}) one has
\bea \label{id1}
x  &=& \frac{Z_2}{Z_0}\\ \label{id2}
y &=&  \frac{Z_1}{Z_0}\\ \label{id3}
\frac{t-1}{g} &=& \frac{Z_3}{Z_0} = xy\\ \label{id4}
\frac{g t^2 + (1-g) t - eg}{gt} &=& \frac{Z_1 + Z_2}{Z_0} = x+y \\ \label{id5}
\Delta &\doteq&\frac{Z_2 - Z_1}{Z_0} = x - y.
\eea
(\ref{id1} - \ref{id2}) show that the image of the elliptic space curve under $\pi_1$ is nothing but the invariant level set in the original coordinates, (\ref{T-invariant}), completed in the projective plane $\mathbb{P}^2$. This image is, for generic values of $E$, a smooth quartic curve in the affine plane but with two double points on the line at infinity at $ [0:1:0]$ and $ [0:0:1]$ respectively. Since the space quartic is, generically, a smooth space curve, it follows that the pull-back $\pi^{-1}_1$ effectively desingularizes the plane quartic; in fact, it amounts to a blow-up of $\mathbb{P}^2$ at 
 $ [0:1:0]$ and $ [0:0:1]$.

On the other hand, (\ref{id3}) shows that $\pi_2$ is projection onto the $t$-line where $t = 1 + gxy$ which presents the space curve as a double cover of $\mathbb{P}^1$. To see this we note from (\ref{id3}) that fixing $t$ determines $\frac{Z_3}{Z_0} = xy$;  then, (\ref{id4}) determines 
$ \frac{Z_1 + Z_2}{Z_0} = x + y$. So then $Z_2/Z_0=x$ and $Z_1/Z_0=y$ are determined up to two choices by the sign of $\Delta$ in (\ref{id5}). This shows that the {\it elliptic involution} of the underlying abstract elliptic curve is realized on the space quartic by the rational involution which interchanges $Z_1$ and $Z_2$ in $\mathbb{P}^3$. (Consistent with this, we note that this involution preserves both $F$ and $G_E$ and hence must induce an involution of the space quartic.) We can make use of this involution to explicitly define the double cover of the $t$-line. 

Recall that the Weierstrass points of an elliptic curve are the fixed points of its elliptic involution.
This is realized in different ways in the various models of the curve that we have been describing. In the Segre (space) model we have just seen that the involution is realized by exchanging $Z_1$ and $Z_2$. Hence the fixed points in this model are the four points of intersection of the space quartic with the plane $Z_1 - Z_2 = 0$. Making the consequent substitutions $Z_2 = Z_1, Z_0 =1, Z_3 = Z_1^2$ in the factored form $G_e$, one reduces to 
\beann
(g Z_1^2 + 1) (g Z_1^2 - 2 Z_1+ 1/g) = e.
\eeann 
 Setting the first factor equal to $t$ and the second factor to $e/t$ and solving for $t$ in terms of $e$ and $g$ one derives the equation for the Weierstrass points in terms of the zeroes of
\bea \label{t-eqn}
\Delta(t) &=& g^2 t^4 - 2g(g+1) t^3 + ((1+g)^2 -2e g^2 ) t^2 + 2eg(g-1) t + e^2g^2.
\eea
More commonly these are referred to as the {\it branch points} of the projection $\pi_2$ from the space quartic onto the $t$-line. We have denoted the function in (\ref{t-eqn}) by $\Delta(t)$ because it is straightforward to check that this quartic is proportional to $\Delta$ defined in (\ref{id5}), consistent with the fact that the elliptic involution is given by interchanging $Z_1$ and $Z_2$. Indeed, the elliptic irrationality corresponds to $\sqrt{\Delta}$.

The elliptic involution in the plane curve model is realized by exchanging $x$ and $y$: 
$(x,y) \to (y,x)$ on the curve. The four Weierstrass points of the invariant curve in the plane are the fixed points determined by setting $x = y$ in (\ref{T-invariant}). So these fixed points are of the form $(x,x)$ where $x$ solves
\bea \label{Weierstrass}
g^2 x^4 - 2g x^3 + (1+g) x^2 - 2 x - E = 0.
\eea
We can relate these points to the zeroes of $\Delta(t)$ which we denote by $t_i, i = 0, \dots, 3$. Using \eqref{id4} these respective coordinates of the Weierstrass points can be related by setting $x = y$:
\beann
x_i = y_i &=& \frac12 \left( t_i + \frac1g - 1 - \frac{e}{t_i}\right),
\eeann
which then must also be the solutions of (\ref{Weierstrass}).

We may further make use of these observations about the elliptic involution to directly parametrize $x$ and $y$ in terms of elliptic functions. The function $Z_i$ restricted to the quartic curve which is the intersection $\{F=0\} \cap\{ G_E=0\}$ is necessarily a fourth order theta function on the curve \cite{FK01}. It has four zeroes on the curve which we can determine by algebra. We illustrate this in the case of $Z_0$. If $Z_0 = 0$, it follows from the form of $F$ that one must also have that either $Z_1 = 0$ or $Z_2 = 0$. It then also follows from the form of $G_e$ that either $Z_3 = 0$ or $gZ_3 - Z_1 - Z_2 = 0$. Hence the points of intersection in $\{Z_0 = 0\} \cap \{ F=0\} \cap\{ G_e=0\}$ are $[0:0:1:0] + [0:1:0:0] + [0:0:g:1] + [0:g:0:1]$ representing them as a formal linear combination of the intersection points.  We call this the divisor of $Z_0$ and denote it by $(Z_0)$. One can similarly work out the divisors of the other $Z_i$. Here are the results: 
\beann
(Z_0) &=& [0:0:1:0] + [0:1:0:0] + [0:0:g:1] + [0:g:0:1]\\
(Z_1) &=& 2 [0:0:1:0] + [0:0:g:1] + [1:0 : -E: 0]\\
(Z_2) &=& 2 [0:1:0:0] + [1 :-E : 0 : 0] + [0:g:0:1]\\
(Z_3) &=& [0:0:1:0] + [0:1:0:0] + [1 : -E : 0 :0] + [1 : 0 : -E :0]
\eeann
From (\ref{id1} - \ref{id2}) we can then also determine the divisors of $x$ and $y$ for which there is significant cancellation:
\bea \label{x-divisor}
(x) &=&  [0:1:0:0] + [1 : -E : 0 :0] - [0:0:1:0] - [0 : 0 : g :1]\\ \label{y-divisor}
(y) &=&  [0:0:1:0]  + [1 : 0 : -E :0] - [0:1:0:0] - [0 : g : 0 :1].
\eea
\medskip 

Now let us pass to the projection $\pi_2$ and consider the uniformizing variable $\xi$ given by the abelian integral of the first kind
\bea \label{abelint}
\xi &=& \int^t \frac{g d \tau}{\sqrt{\Delta(\tau)}}.
\eea
This sets up a mapping
\beann
(x,y) \to (gxy +1, sgn(x-y)) &=&  \left(t, sgn \sqrt{\Delta(t)}\right) \to \xi
\eeann
from a point on the space quartic  to the parameter $\xi$ on the universal cover of the elliptic curve. Modulo the periods of $\xi$, this map is a globally analytic homeomorphism.

By standard results on the Weierstrass sigma function, $\sigma$ \cite{Mar77}, the uniformization (\ref{abelint}) and the divisor information (\ref{x-divisor} - \ref{y-divisor}) it follows that the functions $x$ and $y$ can be paramerized as
\beann
x &=& c_1 \frac{\sigma(\xi -a_1) \sigma(\xi - a_2)}{\sigma(\xi - b_1) \sigma(\xi-b_2)}\\
y &=& c_2 \frac{\sigma(\xi-b_1) \sigma(\xi - a_3)}{\sigma(\xi - a_1) \sigma(\xi-b_3)}
\eeann
where we have the correspondence
\beann
a_1 &\leftrightarrow &  [0:1:0:0] \\
b_1 &\leftrightarrow &  [0:0:1:0] \\
a_2 &\leftrightarrow &  [1 : -E : 0 :0] \\
b_2 &\leftrightarrow &  [0 : 0 : g :1] \\
a_3 &\leftrightarrow &  [1 : 0 : -E :0] \\
b_3 &\leftrightarrow &  [0 : g : 0 :1]
\eeann
with $\xi$-coordinates on the left and points on the space quartic on the right. The $c_i$ are some constants to be determined. We mentioned earlier that the $\pi_1^{-1}$ desingularizes the plane quartic by pulling back to the Segre embedding. We can now make this more precise by observing from the above that it is the two points $a_1$ and $b_3$ that lie over $[0:1:0]$, and 
$b_1$ and $b_2$ that lie over $[0:0:1]$. (We will observe in the next paragraph that $b_1 = - a_1$ and $b_3 = - b_2$ so that these pairs are elliptic involutes of one another just as their base points are.) 

Using the fact that the elliptic involution on the space quartic is given by the interchange of $Z_1$ and $Z_2$, it follows that the following pairs are in involution on the $\xi$-plane
\beann
a_1 && b_1 \\
a_2 && a_3 \\
b_2 && b_3.
\eeann 
But in $\xi$ the elliptic involution is given by sign change, meaning, that 
\beann
b_1 &=& - a_1\\
a_3 &=& - a_2\\
b_3 &=& -b_2
\eeann
so that the sigma function representations reduce to the following fundamental parametrization
\bea \label{param1}
x(\xi) &=& c_1 \frac{\sigma(\xi-a_1) \sigma(\xi - a_2)}{\sigma(\xi + a_1) \sigma(\xi-b_2)}\\ \label{param2}
y(\xi) &=& c_2 \frac{\sigma(\xi + a_1) \sigma(\xi + a_2)}{\sigma(\xi - a_1) \sigma(\xi + b_2)}.
\eea
We also record here, for later use, an application of the fundamental divisor relation for elliptic functions \cite{Mar77}:
\bea \label{divreln}
2a_1 +a_2 - b_2 \equiv 0 \mod \{ 2 \omega_1, 2 \omega_2\}
\eea 
where $2\omega_1, 2\omega_2$ are fundamental periods for the elliptic space curve.

To help pin down the coefficients $c_i$, let us observe from the $\xi - (x,y)$ correspondence above that 
\beann
\xi = a_2 &\leftrightarrow& (x,y) = (0, -E)\\
\xi = - a_2 &\leftrightarrow& (x,y) = (-E, 0).
\eeann
Using the fact that $\sigma$ is an odd function in (\ref{param1} - \ref{param2}) one has
\beann
x(-\xi) &=& c_1 \frac{\sigma(-\xi-a_1) \sigma(-\xi - a_2)}{\sigma(-\xi + a_1) \sigma(-\xi-b_2)}\\
&=& c_1 \frac{\sigma(\xi+a_1) \sigma(\xi + a_2)}{\sigma(\xi - a_1) \sigma(\xi+b_2)}\\
&=& \frac{c_1}{c_2} y(\xi).
\eeann
Then setting $\xi = a_2$ in this relation and using the previous identifications yields
\beann
-E &=& x(-a_2) \\
&=& \frac{c_1}{c_2} y(a_2) = - \frac{c_1}{c_2} E.
\eeann
Hence, $c_1 = c_2 \doteq c$.  It follows that $x(\xi)$ and $y(\xi)$ are in fact elliptic involutes of each other. As a corollary we see that elliptic involution of the level set in the $(x,y)$-plane is induced by the coordinate exchange $(x,y) \to (y,x)$ and that the second component (\ref{projDS2}) of the dynamic map is itself induced by elliptic involution. This calculation shows that
\bea \label{c}
c= -\frac{c E}{x(-a_2)} &=& E \frac{\sigma(a_1 - a_2) \sigma(a_2 + b_2)}{\sigma(a_1 + a_2) \sigma(2 a_2)}.
\eea

Let us now relate the above observations to the parametrization of the map (\ref{projDS1} - \ref{projDS2}). The fact is that this dynamical system expresses an elegant geometric construction.

Consider an initial condition $(x_0, y_0)$ that corresponds to a point $\xi_0$ on the abstract elliptic curve and its first iterate $(x_1, y_1)$ corresponding to another point $\xi_1$ on the abstract elliptic curve. In other words
\beann
(x_0, y_0) &=& (x(\xi_0), y(\xi_0))\\
(x_1, y_1) &=& (x(\xi_1), y(\xi_1))
\eeann
where $x(\xi)$ is defined by (\ref{param1}).

Starting with the initial point one fixes the vertical line $x = x_0$. This line meets the plane quartic (\ref{T-invariant}) in four points, two of which are the (multiplicity 2) double point $[0:1:0]$ at infinity. A third point is, of course, the initial point $(x_0, y_0)$, which we know lies on the plane quartic. To find the remaining point is a matter of algebra; its $y$-coordinate is the other root of
\beann
gx_0(gx_0 -1) y^2 -(gx_0^2 - (1+g)x_0 +1)y - \left( gx_0(gx_0 -1) y_0^2 -(gx_0^2 - (1+g)x_0 +1)y_0\right).
\eeann
After cancellations we find that this root is simply
\beann
\frac{x_0 - 1}{gx_0} - y_0;
\eeann
So the residual point of intersection of the line $x = x_0$ with the invariant curve is
\beann
\left(x_0, \frac{x_0 - 1}{gx_0} - y_0\right).
\eeann
Finally, applying the elliptic involution to this point yields our map
\beann
(x_1, y_1) &=& \left( \frac{x_0 - 1}{gx_0} - y_0, x_0\right).
\eeann 
So, in summary, our dynamic map can be described in completely geometric terms: Starting with an initial point  $(x_0, y_0)$ on the planar invariant curve one finds the residual point of intersection of the line $x=x_0$ with the curve and then flips it across the diagonal $x=y$ to produce  $(x_1, y_1)$. Clearly this process can be arbitrarily iterated. It can also be reversed to yield the inverse map by starting with the vertical line $y = y_0$ and proceeding analogously. 

Hence our map determines an algebro-geometric addition law on each of the affinely smooth level sets of the invariant. (We shall see in the next section that this extends to singular affine level sets as well.) Indeed, this type of addition law has been studied before \cite{G76} and shown to linearize on the abstract elliptic curve. Hence, there is a phase $\nu$ such that $\xi_1 = \xi_0 + \nu$ and the higher iterates $x_n = x(\xi_n)$ are determined by $\xi_n = \xi_0 + n \nu$.
To pin down this phase we apply Abel's theorem, as described in \cite{G76}, which in our setting implies that if $A, B, C, D$ are four points of intersection of a line in the plane with the planar quartic (\ref{T-invariant}) then
\beann
\xi(A) + \xi(B) + \xi(C) + \xi(D) = K
\eeann
where $K$ is a constant  {\it independent} of the planar line one considers. So for our case, in choosing the line $x = x_0$ one has
\beann
\xi_0 + 2 \xi([0:1:0]) -\xi_1 = K.
\eeann
(The factor of 2 here comes from the fact that $[0:1:0]$ is a double point on the curve.)
But in fact, as we observed, this linear relation holds for any iterate and so one has in general
\bea \label{Abel}
\xi_n + 2 \xi([0:1:0]) -\xi_{n+1} = K.
\eea
To pin down the constant $K$ one can consider the line $x =0$ in $\mathbb{P}^2$. Starting with the standard projectivization in coordinates $[x:y;z]$ of (\ref{T-invariant}),
\bea \label{projinvt}
(1+g)xyz^2 - yz(z^2 + gx^2) -xz(z^2 + g y^2) + g^2 x^2 y^2 - E z^4 = 0.
\eea
Considering this in the vicinity of the point $[0:1:0]$ and setting $y=1$, the invariant has the form
\bea \label{doublept}
(1+g)xz^2 - z(z^2 + gx^2) -xz(z^2 + g ) + g^2 x^2 - E z^4 = 0.
\eea
To leading order near $[0:1:0]$ this has the form $0 = gx(gx -z)$ + higher order terms, and so this point is a double point of the curve (as we had already observed) and the two branches of the curve there are $x = 0$ and $z = gx$. Hence the line $x=0$ is tangent to the first branch and the line $x=0$ must meet this double point with multiplicity 3. Indeed, this is confirmed by setting $x=0$ in the LHS of (\ref{doublept}) which evaluates to $-z^3$. There is therefore just one further point of intersection of the line with the affine part of the curve (\ref{projinvt}). By setting 
$z=1$ and $x=0$ to get
\beann
y = -E,
\eeann
one sees that the third point of intersection is $[0:-E: 1]$ which corresponds to $a_2$ in the Segre model.
The points at infinity in the intersection of $x=0$ with the curve correspond to $2 a_1$, because of the tangency,  and $b_3$. It then follows from Abel's theorem that
\beann
2 a_1 + a_2 + b_3 &=& K\\
2 a_1 + a_2 -  b_2 &=& K\\
0 &\equiv & K.
\eeann
where the second line follows from the involution pairings and the third line follows from (\ref{divreln}). Therefore the RHS of (\ref{Abel}) vanishes modulo periods.

Putting all this together one has
\bea \label{phase}
\nu = \xi_{n+1} - \xi_n &\equiv& 2 \xi([0:1:0]) \\ \nonumber
&=& a_1 + b_3 \\ \nonumber
&=& a_1 - b_2 \\ \nonumber
&\equiv & -(a_1 + a_2) ,
\eea
which is clearly independent of $n$. The evaluations in the second and third equalities are made using the definitions of the $a_j, b_j$ and the relations among them are determined by involution and (\ref{divreln}). 
\bigskip

Finally, we are in a position to present an elliptic parametrization of the map (\ref{projDS1} - \ref{projDS2}):
\bea \label{addparam1}
x(\xi + \nu) &=& \frac1g - \frac1g  \frac1c \frac{\sigma(\xi+a_1) \sigma(\xi - b_2)}{\sigma(\xi - a_1) \sigma(\xi-a_2)} 
- c \frac{\sigma(\xi + a_1) \sigma(\xi + a_2)}{\sigma(\xi - a_1) \sigma(\xi + b_2)}\\ \label{addparam2}
y(\xi + \nu) &=& c \frac{\sigma(\xi-a_1) \sigma(\xi - a_2)}{\sigma(\xi + a_1) \sigma(\xi-b_2)}.
\eea
Making use of (\ref{param1}) in (\ref{addparam1}) we see that the poles on the LHS of (\ref{addparam1}) are located at $\xi = -\nu - a_1 = a_2$ and $\xi = -\nu + b_2 = 3 a_1 +2 a_2$ while those on the RHS must occur among the values $\xi = a_1, \xi = a_2$ and $\xi = - b_2$. (One of these latter values cannot occur as a pole in order to maintain the required balance between the number of poles on the two sides of equation (\ref{addparam1}); i.e. upon passing to a common denominator on the LHS, one of the zeroes of the numerator must cancel one these possible polar values.) So upon comparison we see that $a_2$ is a pole. On the other hand $a_1$ must be the potential pole on the RHS that gets cancelled, since otherwise we would deduce that $a_2 = - a_1$ which we know from our earlier analysis is not the case. Hence we may deduce that
\bea  \nonumber
3 a_1 + 2 a_2 &=& - b_2 \\ \nonumber 
3 a_1 + 2 a_2 &=& -2 a_1 - a_2 \\ \label{balanceformula}
5 a_1 + 3 a_2 &=& 0.
\eea

\subsection{Separatrix Degeneration to the Combinatorial Orbit}
We now specialize our considerations to the case of the invariant level curve which passes through the unique hyperbolic fixed point of the map (\ref{TMap}) which has energy level $e_1$. This energy is given as a function of $g$ in Table \ref{table:DegenerateEnergies}:
\beann
e &=& \dfrac{1 + 20 g - 8g^2 + (8g-1)\sqrt{1-8g}}{32g^2}
\eeann
which we will continue to denote by $e$ in this section. Setting $g=\dfrac{1- \gamma^2}{8}$ we may express this and other parameters relevant to the degenerate level set purely in terms of $\gamma$:
\bea
g &=& \frac{1-\gamma^2}{8}\\
\gamma &=& \sqrt{1 - 8 g}\\ \label{E-gamma}
e &=& \frac{(\gamma + 3)^3}{4(1 + \gamma)(1 - \gamma^2)}\\ 
x_* & = & \frac{1 - \sqrt{1 -8g}}{4g}\\
&=& \frac2{1 + \gamma}
\eea
where $x_*$ denotes the fixed point, (\ref{fixedpt}), of the separatrix. Given these definitions we also have a concise expression for the eigenvalues of the linearized combinatorial orbit at $x_*$
\bea \label{g-eigenvalue}
\chi^{\pm 1} &=& \frac{1 \pm \sqrt{1 - 4 g^2 x^4_*}}{2gx^2_*}\\ \label{eigenvalue}
&=& \left(\frac{1 + \sqrt{\gamma}}{1 - \sqrt{\gamma}}\right)^{\pm 1}.
\eea
In addition, the $t$-equation degenerates to
\beann
0 &=& [2(1-\gamma)t - (3-\gamma)]^2 
[4 (\gamma^2 - 1)^2 t^2 -4(1-\gamma^2) (\gamma - 3) (\gamma + 5)t + (\gamma - 3)^4]
\eeann
whose roots are
\beann
t_0, t_1 &=& \frac{\gamma - 3}{2 (\gamma^2-1)} \left( 5 + \gamma \pm 4\sqrt{ (1 + \gamma)} \right)\\
t_2 = t_3 &=&  \frac{3-\gamma}{2 (1 - \gamma)} = \frac12 \left( 1 + \frac2{1 - \gamma}\right).
\eeann
By completing the square in the quadratic factor, the equation for $t$ may be rewritten as 
\bea
0 &=& \left(z + r_1\right)^2 
\left(z^2 - r_2^2\right)\\
r_1 &=&  \frac6{1 - \gamma^2}\\
r^2_2 &=& 4\frac{(\gamma - 3)^2}{(\gamma - 1)^2 (1+\gamma)}\\
\sqrt{r_1^2 - r_2^2} &=& \frac{2 \sqrt{-\gamma(\gamma^2 -5\gamma +3)}}{1 - \gamma^2}\\ \label{zparam}
z &=& t - \frac12 \frac{(3-\gamma)(5 + \gamma)}{1 - \gamma^2}.
\eea
Here $\sigma_d$ denotes the generalized sigma function that $\sigma$ degenerates to on the singular curve that (\ref{T-invariant}) limits to. We need to determine the form of $\sigma_d$. We observe that the equation (\ref{Weierstrass}), determining the Weierstrass points, degenerates in this limit to
\beann
((1 + \gamma) x - 2)^2 \left((1-\gamma)^2 x^2 - 4(1 - \gamma^2) (\gamma + 3) x + 4(\gamma^2 + 10 \gamma + 5)\right) = 0.
\eeann 
The sigma function may be explicitly expressed in terms of the odd Jacobi theta function $\vartheta_1$ \cite{Mar77} as
\beann 
\sigma(\xi) &=& 2 \omega_1 \dfrac{\vartheta_1(v)}{\vartheta'_1(0)} e^{2 \omega_1 \eta_1 v^2}
\eeann
where $v= \frac{\xi}{2\omega_1}$ and $\vartheta_1$ depends on the modular parameter
$\tau = \omega_2/\omega_1$ where $\omega_1$ is real.  $\eta_1$ is an integration constant of the Weierstrass $\zeta$-function. The separatrix limit here in which $t_2$ and $t_3$ coalesce (or equivalently $x(\omega_2)$ and $x(\omega_3)$ coalesce) corresponds to the infinite period limit in which $\tau \to \infty$. In this limit, $\vartheta_1$ must converge to a linear combination of exponentials of the form $e^{s \xi}$ and $e^{-s \xi}$ which is odd since $\vartheta_1$ remains odd throughout the limit. Hence in fact it must be a multiple of $\sinh(s \xi)$. Then, setting $s = \lim \frac{1}{2\omega_1}$ and $\rho = \lim \frac{\eta_1}{2\omega_1}$, the limiting form is 
\bea \label{degsigma}
\sigma_d(\xi) &=& e^{\rho \xi^2} \frac{\sinh(s \xi)}{s}.
\eea 
One has the implicit integral inversion:
\bea \label{integral}
\xi  &=& \int\frac{dz}{(z + r_1)\sqrt{z^2 - r_2^2}}.
\eea
Setting $z = r_2 \cosh(x)$ this transforms to 
\beann
\xi &=& \int \frac{dx}{r_2 \cosh(x) + r_1}\\
&=& \frac2{r_2}\int \frac{d u}{u^2 + \frac{2r_1}{r_2} u +1}\\
&=& \frac2{ \sqrt{r_2^2 - r_1^2}}\tan^{-1} \frac{r_2 u + r_1}{\sqrt{r_2^2 - r_1^2}} + \alpha
\eeann
where $u = e^x$. Unravelling the substitutions and setting $z = r_2 w$ yields
\beann
\sqrt{1 - (r_1/r_2)^2} \tan\left(\frac{\sqrt{r_2^2 - r_1^2}}{2} (\xi - \alpha) \right) 
&=& (w + r_1/r_2) \pm \sqrt{w^2 - 1} \\
\sqrt{(r_1/r_2)^2 - 1} \tanh \left(\frac{\sqrt{r_1^2 - r_2^2}}{2} (\xi - \alpha) \right) - r_1/r_2
&=& w \pm \sqrt{w^2 - 1}\\
\frac{\sqrt{-\gamma(\gamma^2 -5\gamma +3)}}{(3-\gamma) \sqrt{1 + \gamma}} \tanh \left(\frac{ \sqrt{-\gamma(\gamma^2 -5\gamma +3)}}{1 - \gamma^2} (\xi - \alpha) \right) - \frac{3}{(\gamma-3)\sqrt{1 + \gamma}} &=&w \pm \sqrt{w^2 - 1}\\
\frac{1}{(3-\gamma)\sqrt{1 + \gamma}} \left(\sqrt{-\gamma(\gamma^2 -5\gamma +3)} \tanh \left( \frac{ \sqrt{-\gamma(\gamma^2 -5\gamma +3)}}{1 - \gamma^2} (\xi - \alpha) \right) +3\right) &=&w \pm \sqrt{w^2 - 1}
\eeann
This last expression presents $\xi$ as a {\it closed form} branched 2:1 cover of the $w$ plane or, by scaling and translation, the $t$-plane. We can also pin down the values at branch points:
\beann
z = - r_1 \iff  \xi - \alpha &=& \pm \infty \\
z = \pm r_2 \iff  \xi - \alpha  &= &\frac2{\sqrt{r_1^2 - r_2^2}}\tanh^{-1} \left(\frac{r_1 + r_2 }{r_1 - r_2}\right)^{\pm 1/2}.\\
\eeann
We can apply the above inversion analysis to the near-separatix linearization. Setting 
$z = -r_1 + q$ in (\ref{integral}) gives
\beann
\xi &=& \int_\epsilon \frac{dq}{q\sqrt{(r_1 - q)^2 - r_2^2}} \\
- \sqrt{r_1^2 - r_2^2} \xi &=& \log \epsilon (1 + \mathcal{O}(\epsilon/ \log \epsilon))\\
\epsilon &=& e^{- \sqrt{r_1^2 - r_2^2} \xi} (1 + o(1)).
\eeann
Comparing to the inversion of (\ref{integral}) detailed above, one may infer that
\bea
s &=& \frac12 \sqrt{r_1^2 - r_2^2}
\eea
and relate this to the eigenvalues at the linearization at the separatrix: the linearization  realized in (\ref{phase}) presents the dynamics as a phase increment on the covering space of the planar quartic. This fact extends to the degenerate limit of the separatrix. This flow is the tangent space flow near the fixed point and so must coincide with the linearized flow in the directions of the stable or the unstable manifolds. But these flows, we know, are just the power flows of the respective eigenvalues, which are the exponentials of these phase translations. It follows  that
\bea \label{chi}
\chi^\pm &=& e^{\mp 2 s \nu}.
\eea
\medskip

One can now say more about the dynamical system (\ref{projDS1} - \ref{projDS2}) on the separatrix level set. The parametrization on this level set becomes
\beann
x(\xi ) &=& c \,\, e^{\rho (a_2^2 - b_2^2)} e^{-2\rho(2a_1 -b_2+a_2)\xi}
\frac{\sinh(s(\xi - a_1)) \sinh(s(\xi - a_2))}{\sinh(s(\xi  + a_1)) \sinh(s(\xi - b_2))}\\ \nonumber
&=& c \,\, e^{\rho(a_2^2 - b_2^2)} e^{-2\rho(2a_1 -b_2+a_2)\xi} e^{- s (2 a_1 + a_2 - b_2)}
\frac{\left( 1 - e^{-2s(\xi - a_1)}\right) \left( 1 - e^{-2s(\xi - a_2)}\right)}{\left( 1 - e^{-2s(\xi + a_1)}\right) \left( 1 - e^{-2s(\xi - b_2)}\right))}\\
y(\xi) &=& x(-\xi)
\eeann
We note that $2a_1 -b_2+a_2$ is the limit of divisors of an elliptic function as noted in (\ref{divreln}) and therefore must equal the limit of a linear combination of periods. By re-centering the origin of $\xi$ we may assume this sum is zero. Independently, this is also verified from the dynamical systems perspective since, otherwise,  $x$ would blow up or vanish as $n$ goes to infinity which our phase plane analysis has shown is not the case.  (Indeed, this argument shows that the real part of $2a_1 -b_2+a_2$ must already be zero, before any centering.) Hence, the expression for the first component reduces (using $b_2 = 2a_1 + a_2$) to 
\beann 
x(\xi) &=& c \,\, e^{- 4\rho a_1 (a_1 + a_2)} 
\frac{\left( 1 - (e^{-2s\nu})^{\frac{\xi - a_1}{\nu}}\right) \left( 1 - (e^{-2s\nu})^{ \frac{\xi - a_2}{\nu}}\right)}{\left( 1 - (e^{-2s\nu})^{\frac{\xi + a_1}{\nu}}\right) \left( 1 - (e^{-2s\nu})^{ \frac{\xi - 2 a_1 - a_2}{\nu}}\right)}. \,\,\,\,\,\,\,\,\,\,\,\,\,\,\,\,
\eeann 
In the separatrix limit, as $\xi \to \infty$, we must have 
\bea \label{pre-ev}
c e^{- 4\rho a_1 (a_1 + a_2)} & = & x_*.
\eea
Hence, one finally has
\bea \label{degenorbit}
x(\xi) &=& x_*
\frac{\left( 1 - (e^{-2s\nu})^{\frac{\xi - a_1}{\nu}}\right) \left( 1 - (e^{-2s\nu})^{ \frac{\xi - a_2}{\nu}}\right)}{\left( 1 - (e^{-2s\nu})^{\frac{\xi + a_1}{\nu}}\right) \left( 1 - (e^{-2s\nu})^{ \frac{\xi - 2 a_1 - a_2}{\nu}}\right)} \\
y(\xi) &=& x(-\xi).
\eea
To establish (\ref{eq: T_j closed form}) it remains to impose the initial condition (\ref{eq: Tj rec derivation ic}). In our current framework this means that we need to find $\xi_0$ such that
$x(\xi_0 - \nu) = 0$. This is achieved in terms of (\ref{degenorbit}) by setting  
\beann
\xi_0 - \nu - a_1 &=& 0\\
\xi_0 + a_1 + a_2 - a_1 &=& 0\\
\xi_0 &=& - a_2.
\eeann
With this in place, the combinatorial orbit is given by
\beann
x(\xi_0 + n \nu) &=& x_*
\frac{\left( 1 - (e^{-2s\nu})^{n+1}\right) \left( 1 - (e^{-2s\nu})^{ n +\frac{2 a_2}{a_1 + a_2}}\right)}{\left( 1 - (e^{-2s\nu})^{n+\frac{a_2 - a_1}{a_1 + a_2}}\right) \left( 1 - (e^{-2s\nu})^{n + 2}\right)} \\
&=& x_*
\frac{\left( 1 - (e^{-2s\nu})^{n+1}\right) \left( 1 - (e^{-2s\nu})^{ n +\frac{a_2 - a_1}{a_1 + a_2} + 1}\right)}{\left( 1 - (e^{-2s\nu})^{n+\frac{a_2 - a_1}{a_1 + a_2}}\right) \left( 1 - (e^{-2s\nu})^{n + 2}\right)}\\
&=& x_*
\frac{\left( 1 - (e^{-2s\nu})^{n+1}\right) \left( 1 - (e^{-2s\nu})^{ n +5}\right)}{\left( 1 - (e^{-2s\nu})^{n+4}\right) \left( 1 - (e^{-2s\nu})^{n + 2}\right)}
\eeann
where the last line is justified by (\ref{balanceformula}) which persists under degeneration because Abel's theorem does. In particular one sees that $\frac{a_2 - a_1}{a_1 + a_2} = 4$ is equivalent to $5 a_1 + 3 a_2 = 0$. 

Finally we need to relate this to Bousquet-Melou's $Z$ in (\ref{eq: T_j closed form}). Comparing to (\ref{chi}) one has $\chi = \chi^+ = e^{-2s\nu}$ and so one may rewrite the combinatorial flow as
\beann
x(\xi_0 + n \nu) &=& x_*
\frac{\left( 1 - \chi^{n+1}\right) \left( 1 - \chi^{ n +5}\right)}{\left( 1 - \chi^{n+4}\right) \left( 1 - \chi^{n + 2}\right)}.
\eeann
It then suffices to show that the eigenvalue $\chi$ satisfies the equation (\ref{Zeqn}) for $Z$.
Starting with (\ref{g-eigenvalue}), and using the identity $2gx^2_* = x_* -1$ for the fixed point one has
\beann
\chi^\pm &=& \frac{1 \pm \sqrt{1 - (x_* - 1)^2}}{x_* -1}
\eeann
which is easily inverted by appropriate squaring to yield
\beann
x_* &=& \frac{(\chi + 1)^2}{\chi^2 +1}
\eeann
which finally establishes (\ref{eq: T_j closed form}).

We close with a summary of the key results in this section:

\begin{theorem}\label{thm: main theorem c=0}
The discrete integrable recurrence
\begin{equation}
x_n = 1 + gx_n(x_{n-1} + x_{n+1})
\end{equation}
with invariant
\begin{equation}
I(x_{n}, x_{n-1}) = x_nx_{n-1}\left(1 - g x_n\right) \left( 1 - g x_{n-1}\right) + g x_n x_{n-1} 
- x_n - x_{n-1}+ \frac{1}{g}
\end{equation}
and limit $\lim_{n\rightarrow\infty} x_n = x_-^*$ satisfying $x_-^* = 1 + 2gx_-^{*2}$ possesses a solution
\begin{equation}
x_n = x(\xi_0 + n\nu) = x_*
\frac{\left( 1 - \chi^{n+1}\right) \left( 1 - \chi^{ n +5}\right)}{\left( 1 - \chi^{n+4}\right) \left( 1 - \chi^{n + 2}\right)},
\end{equation} 
where $\nu = 2\xi([0:1:0])$ under the composite Abel map determined by (\ref{integral}), $\pi_1$ and $\pi_2$.  Here $\chi = (\frac{1 + \sqrt{\gamma}}{1 - \sqrt{\gamma}}) = e^{-2s\nu}$ where $\gamma = \sqrt{1 - 8g}$ and $s = \frac12 \sqrt{r_1^2 - r_2^2}$. $\xi_0$ is determined by the initial condition $x(\xi_0 - \nu) = 0$.

Finally, making the combinatorial identification $T_n(g) = x_n$ and $T(g) = x^*_-$ one has 

\begin{equation}
T_n = T \frac{(1-\chi^{n+1})(1-\chi^{n+5})}{(1-\chi^{n+4})(1-\chi^{n+2})}
\end{equation}
for $\chi$ satisfying $T = \frac{(1 + \chi)^2}{(1 + \chi^2)}$.
\end{theorem}

\section{Concluding Remarks} \label{Conclusions}
In this last section we mention a few directions that are currently under investigation for building on what was derived in this paper.

\subsection{Combinatorial Problems in Non-Autonomous Extensions} \label{conclusions1}
The continuous Painlev\'e equations that give rise to Painlev\'e's famous six transcendents are non-autonomous differential equations. In the {\it autonomous limits} of these equations the transcendent limits to a classical function such as the elliptic functions discussed in this paper.

The discrete Painlev\'e equations also have natural non-autonomous extensions. In the case of 
dPI, we consider such an extension of \eqref{eq: discrete dynamics general c}, with $c=1$ having the general form
\bea \label{non-aut-dPI}
x_{n+1} + x_{n-1} &=& \frac{n}N \frac1{g x_n} - \frac{\zeta}{g} - x_n,
\eea 
$g, N$ and $\zeta$ are parameters.  Setting $\zeta = -1 = \frac{n}N$ one recovers our autonomous dPI equation. One may now ask, are there orbits of this non-autonomous system having combinatorial significance related to what we saw in the autonomous case? 
The answer is yes and it comes from the analysis of the asymptotics of orthogonal polynomials and their relation to the enumeration of quadrangulations of surfaces of general genus. (The combinatorial problem studied in the autonomous case was for just planar maps; i.e., maps of genus 0.) This connection is mediated by the analysis of the Riemann-Hilbert problem for orthogonal polynomials as it relates to hermitian random matrix models, carried out in \cite{EM03, EMP08, Er11}.  We briefly outline the essentials of this as it relates to \eqref{non-aut-dPI}. 

Define the orthogonal polynomials $p_n(\lambda) = \gamma_n \lambda^n + \cdots$ for positive $\gamma_n$ with respect to the exponential weight 
\beann
w_\zeta(\lambda) &=& \exp\left( - N \left( \frac{\zeta}{2} \lambda^2 + \frac{g}4 \lambda^4\right)\right),
\eeann
for $g > 0$, so that
\beann
\int p_n(\lambda) p_m(\lambda) w_\zeta(\lambda) d\lambda = \delta_{nm}
\eeann
for $n,m \geq 0.$ These polynomials satisfy the three-term recurrence relation
\beann
\lambda p_{n,N}(\lambda) = b_{n+1, N} p_{n+1}(\lambda) + a_{n,N} p_n(\lambda) + b_{n,N} p_{n-1}(\lambda)
\eeann
for $n \geq 0$ and $p_{-1} = 0$. Since the weight is even, the recurrence coefficients $a_{n,N}$  are all zero, and the polynomials are entirely defined by the recurrence coefficients $b_{n,N}$. These coefficients satisfy a {\it nonlinear} recurrence of their own: 
\bea \label{Freud}
g b^2_{n,N} \left( b^2_{n+1, N} + b^2_{n, N} + b^2_{n-1, N}\right) + \zeta b^2_{n, N} 
= \frac{n}{N}.
\eea
In approximation theory \eqref{Freud} is referred to as Freud's equation.
It is straightforward to see that this coincides with \eqref{non-aut-dPI} if one sets $x_n = b^2_{n, N}$. The initial conditions that then relate these orthogonal polynomials recurrence coefficients to an orbit within our non-autonomous dPI system, \eqref{non-aut-dPI}, are 
$x_0 = 0$ and $x_1 = \frac{c_2}{c_0}$ where
\beann
c_j &=& \int \lambda^j w_\zeta (\lambda) d\lambda
\eeann 
is the $j^{th}$ moment of the measure. The latter is derived from the first equation in the recurrence relations,
\beann
\lambda p_0(\lambda) &=& b_1 p_1(\lambda) + b_0 p_{-1}(\lambda)\\
\lambda \gamma_0 &=& b_1 p_1(\lambda)
\eeann
by the normalization requiremnt  that $p_1(\lambda) = \frac{\lambda \gamma_0}{b_1}$ has norm 1:
\beann
\int p_1(\lambda) p_1(\lambda) w_\zeta(\lambda) d\lambda &=& \int \left( \frac{\lambda \gamma_0}{b_1}\right)^2 w_\zeta(\lambda) d\lambda \\
&=& \frac{\gamma^2_0}{b^2_1} \int \lambda^2  w_\zeta(\lambda) d\lambda\\
&=& 1,
\eeann
so that
\beann
x_1 = b^2_1 &=& \gamma^2_0 \int \lambda^2  w_\zeta(\lambda) d\lambda.
\eeann
Similarly, it is required that $p_0(\lambda) \gamma_0$ have norm 1:
\beann
\int p_0(\lambda) p_0(\lambda) w_\zeta(\lambda) d\lambda 
&=& \gamma^2_0 \int  w_\zeta(\lambda) d\lambda\\
&=& 1.
\eeann
So finally one has the second initial condition:
\beann
x_1 &=& \frac{\int  \lambda^2 w_\zeta(\lambda) d\lambda}{\int  w_\zeta(\lambda) d\lambda}.
\eeann

The connection to the combinatorial problem of enumerating quadrangulations now stems from a result in \cite{EMP08} where it is shown that for $N$ large and with $X \doteq \frac{n}{N} \sim 1$, one has a full asymptotic expansion of $x_n = b^2_{n,N}$:
\beann
x_n \sim X \left( z_0(s) + \frac1{n^2} z_1(s) + \frac1{n^4} z_2(s) + \dots \right)
\eeann
where $s = -X\frac{g}4$. This expansion is uniformly valid on compact subsets of complex $s$ with $\Re s < 0$. The coefficients $z_g(s)$ are the generating functions for 4-valent maps (whose dual maps are the quadrangulations in question) of genus $g$. 

It will be of interest to determine how the combinatorial features just described are related to the 
dynamic properties of the orbit, corresponding to the orthogonal polynomial recurrence relations, of non-autonomous dPI. It will also be interesting to see if this can be related to the analysis of the autonomous case discussed in this paper. There  is already one indication that such relations do hold: the generating function, $z_0(s)$, which enumerates planar 4-valent maps solves the same functional equation, \eqref{eq: R functional equation ch2}, as does the limit $R(g)$ of the distance generating funcitons $R_n(g)$. Thus the non-autonomous $x_n$ and the autonomous $R_n$ agree at leading order in large $n$. These questions are currently under further investigation and will be reported on elsewhere. 

As indicated at the start of section \ref{triangulations} there is an analogous relation between the combinatorics of (Eulerian) triangulations and the $T_j$ generating functions \cite{BDG03a}.  The connection to recurrence relations for (generalized) orthogonal polynomials is more complicated on several levels in this case. For one thing, the coefficients $a_n$ will now no longer automatically be zero. However, the analysis in \cite{EP12} has shown how to handle this additional complication. Based on this we are currently exploring non-autonomous extensions of the $c = 0$ system. 

\subsection{Elliptic Combinatorics} \label{conclusions2}
The combinatorial focus of this paper has been on the separatrix orbit of \eqref{eq: discrete dynamics general c}. Is there a related combinatorial significance for the other, generic, elliptically parametrized orbits? One affirmative answer to this question has been provided in \cite{BDG03a}. There the study of random embedded  trees, corresponding to the sytem  \eqref{eq: discrete dynamics general c} with $c=0$ but with the added structure of {\it walls} is discussed. A wall introduces a conditioning on the random system that strictly bounds the size of the labels/ positions that the random process can attain. A ``one-wall" case which is that the half-line bounded below at 0 or -1  reproduces the model we have been studying in this paper (introduced  in section \ref{triangulations}) that corresponds to the separatrix orbit.  By contrast, a``two-wall" condiitoning would require labels to take their values in a finite subinterval of $\mathbb{Z}$; i.e., it would replace the asymptotic boundary condition (\ref{eq: Tj rec derivation ic} b) by another finite boundary condition. From the dynamical point of view this corresponds to a two-point {\it finite} boundary value problem. The one-wall case we studied in this paper corresponded to a two-point {\it semi-infinite} boundary value problem. 
One may derive recursive formulas for the generating functions of such trees in the general two point boundary value problem of a two-wall conditioning. The effect of this is to select one of the generic orbits we described in section \ref{e-param0} whose closed form solutions for the generating functions will now be in terms of elliptic functions. These generating functions may be used to study the continuum scaling limit in terms of the periods of the associated elliptic curve, which may be applied to the  probabilistic analysis of population spreading. We are studying how the geometric analysis developed in section \ref{e-param0} might be used in this application.


\end{document}